\documentclass[a4paper,UKenglish,cleveref, autoref, thm-restate]{lipics-v2021}

\usepackage{algpseudocode}
\usepackage{xr}
\usepackage[linesnumbered,lined,boxed,commentsnumbered]{algorithm2e}
\usepackage{indentfirst}
\usepackage{algorithm2e}
\RestyleAlgo{ruled}
\usepackage{mathtools}

\title{$m$-Eternal Dominating Set Problem on Subclasses of Chordal Graphs}

\author{Ashutosh Rai}{Department of Mathematics, IIT Delhi}{ashutosh.rai@maths.iitd.ac.in}{}{}
\author{Soumyashree Rana}{Department of Mathematics, IIT Delhi}{maz218122@maths.iitd.ac.in}{}{}

\authorrunning{A. Rai and S. Rana} 

\Copyright{Ashutosh Rai and Soumyashree Rana} 



\category{} 

\nolinenumbers 

\usepackage{color}

\begin{document}

\maketitle              
\begin{abstract}
A \textit{dominating set} of a graph $ G (V, E)$ is a set of vertices $D\subseteq V$ such that every vertex in $ V\setminus D$ has a neighbor in $D$. An eternal dominating set extends this concept by placing mobile guards on the vertices of $D$. In response to an infinite sequence of attacks on unoccupied vertices, a guard can move to the attacked vertex from an adjacent position, ensuring that the new guards' configuration remains a dominating set. In the one guard move model, only one guard moves per attack, while in the all guards move model, multiple guards may move. The set of vertices representing the initial configuration of guards in one guard move model is the \emph{eternal dominating set} and in all guards move models is the \emph{$m$-eternal dominating set} of $G$. The minimum size of such a set in one guard move model and all guards move model is called the \emph{eternal domination number} ($\gamma^\infty(G)$) and the \emph{$m$-eternal domination number} ($\gamma_m^\infty(G)$) of $G$, respectively. Given a graph $G$ and an integer $k$, the  \textsc{$m$-Eternal Dominating Set} asks whether $G$ has an $m$-eternal dominating set of size at most $k$.

In this work, our study focuses mainly on the computational complexity of \textsc{$m$-Eternal Dominating Set} in subclasses of chordal graphs, namely split graphs and undirected path graphs. For the class of split graphs, we show a dichotomy result by first designing a polynomial-time algorithm for $K_{1,t}$-free split graphs with $t\le 4$, and then proving that the problem becomes \textsf{NP}-complete for $t\ge 5$. Further, we establish that the problem is \textsf{NP}-hard on undirected path graphs. Moreover, we exhibit the computational complexity difference between the variants by showing the existence of two graph classes such that, in one, both \textsc{Dominating Set} and \textsc{$m$-Eternal Dominating Set} are solvable in polynomial time while \textsc{Eternal Dominating Set} is \textsf{NP}-hard, whereas in the other, \textsc{Eternal Dominating Set} is solvable in polynomial time and both  \textsc{Dominating Set} and \textsc{$m$-Eternal Dominating Set} are \textsf{NP}-hard. Finally, we present a graph class where \textsc{Dominating Set} is \textsf{NP}-hard, but \textsc{$m$-Eternal Dominating Set} is efficiently solvable. 

\end{abstract}    

\section{Introduction}
Let $G=(V, E)$ be a finite, simple, and undirected graph. A set of vertices $D \subseteq V$ is said to be a \emph{dominating set} of $G$ if each vertex in $V\setminus D$ has an adjacent vertex in $D.$ The minimum cardinality among all dominating sets of $G$ is the \emph{domination number} of $G,$ denoted by $\gamma(G).$ The decision version of the corresponding problem is defined as follows.

  \vspace{3mm}
	\noindent\fbox{
		\begin{minipage}{.95\textwidth}
			 \underline{\textsc{Minimum Domination} (\textsc{Dominating Set})}\\ 
			\textbf{Input:} A graph $G$ and an integer $k\ge 0$.\\ \textbf{Question:} Is $\gamma(G)\le k?$
          \end{minipage}}
  \vspace{2mm}

The domination problem is a fundamental topic in graph theory and has a wide range of applications in network security, facility location, and social network analysis \cite{gupta2013domination,qiang2021novel}. Despite being \textsf{NP}-hard in general \cite{garey2002computers}, and even on restricted graph classes, it has been widely studied for its theoretical depth and practical relevance \cite{cockayne2006domination,haynes2013fundamentals}.

Over the years, numerous variants of domination have been proposed to model more dynamic or constrained settings. The dynamic variants of domination have gained attention due to their applicability in scenarios where the threat or demand is not static, but evolves. More specifically, to address surveillance and protection problems such as military defense strategies \cite{klostermeyer2016protecting} and epidemiology \cite{mahadevan2020application}, where it helps in optimally deploying mobile resources to control disease spread over networks. One such family of problems includes the \emph{eternal domination} problem, which models a setting where a team of guards must defend a graph against an infinite sequence of attacks, one attack at a time.

There are many variants of the \emph{eternal domination} problem. Each version depends on the two parameters: (i) $x$, the number of guards that can move to their neighboring vertex during a defense, and (ii) $y$, the number of guards that can occupy a vertex initially or after a defense, where $x,y\in \mathbb{N}$. The optimal number of guards required in such a variant is denoted by $\gamma_{x,y}^\infty$. When $y=1$, then $\gamma_{x,y}^\infty$ is written as $\gamma_x^\infty$, and if $x$ refers to all the guards being able to move, then it is written as $\gamma_m^\infty$. The subscript `$m$' stands for multiple (all) guards can move in response to an attack, and the superscript `$\infty$' represents an infinite sequence of attacks. In this paper, we focus on the ``\emph{all guard move}'' model, where all guards are allowed to move during a defense, but still, no more than one guard can occupy a vertex. This variant is known as the \emph{$m$-eternal dominating set} problem. 
  
Before defining this model formally, we first introduce the notion of a \emph{guards move}. Suppose we have two dominating sets of a graph $G$, say $D=\{x_1,x_2,\ldots,x_k\}$ and $D'=\{x_1',x_2',\ldots,x_k'\}$, such that for each $i\in [k]$, $x_i\in N[x_i']$. Then $D'$ is said to be obtained from $D$ by a \emph{guards move}. When $x_i\in N(x_i'),$ then it means the guard at $x_i$ has moved to $x_i'$. Now, we define the \emph{$m$-eternal dominating set}. For a graph $G=(V, E)$, a set $D\subseteq V$ is said to be an \emph{m-eternal dominating set}, if for any infinite sequence of attacks $R=(r_1,r_2,\ldots)$, $r_i\in V$, $i=1,2,\ldots$, there exists a sequence of dominating sets $(D_0, D_1, D_2,\ldots)$ such that $D=D_0$, $D_i$ is obtained from $D_{i-1}$ by a \emph{guards move}, and $r_i\in D_i$, for each $i\ge 1$. Here, $D_i$ represents the configuration of guards after defending against the $i^{th}$ attack. The minimum number of guards required to form an m-eternal dominating set is called the \emph{m-eternal domination number} of $G$, denoted by $\gamma_m^\infty(G)$. The decision version of the corresponding problems for both models is defined as follows.
  \vspace{3mm}
	\noindent\fbox{
		\begin{minipage}{.95\textwidth}
			 \underline{\textsc{Minimum} $m$-\textsc{Eternal Domination} ($m$-\textsc{Eternal Dominating Set})}\\ 
			\textbf{Input:} A graph $G$ and  an integer $k\ge 0$.\\ 
            \textbf{Question:} Is $\gamma_m^\infty (G)\le k?$
          \end{minipage}}\\
	\noindent\fbox{
		\begin{minipage}{.95\textwidth}
			 \underline{\textsc{Minimum Eternal Domination} (\textsc{Eternal Dominating Set})}\\ 
			\textbf{Input:} A graph $G$ and  an integer $k\ge 0$.\\ 
            \textbf{Question:} Is $\gamma^\infty (G)\le k?$
          \end{minipage}}

\textsc{Eternal Dominating Set} is closely related to other graph dynamic protection problems, particularly the \emph{eternal vertex cover} problem. In this variant, guards are placed on the vertices to protect edges instead of the vertices. When an edge is attacked, a guard can move along that edge to defend it, and after the move, the guard's new configuration must form a valid vertex cover. Similar to \emph{eternal domination}, the model can vary depending on whether one guard or multiple guards are allowed to move in response to an attack.
A recent studied variant is the \emph{connected eternal vertex cover}, where the set of occupied vertices must induce a connected subgraph. This is particularly useful in scenarios such as surveillance or communication networks, as it requires continuous communication among the guards.
Despite the \textsf{NP}-hardness of this problem, the eternal vertex cover problem and its variants have received increasing attention in algorithmic graph theory (see papers \cite{paul2023some}, \cite{babu2022eternal}, \cite{fujito2020eternal}, \cite{paul2024eternal}).
\vspace{-15pt}
  \subsection{Related work}
\vspace{-5pt}
  The \textsc{Eternal Dominating Set} was initially introduced as infinite order domination by Burger et al. \cite{burger2004infinite} in 2004. It models a dynamic defense strategy on graphs, where a team of guards must respond to an unending sequence of vertex attacks while always maintaining a dominating set. In their work, they computed the eternal domination number for several fundamental graph classes such as paths, cycles, and complete multipartite graphs. Later that same year, Goddard et al. \cite{goddard2005eternal} formalized this as the eternal security problem and introduced two key models: the one guard move model and the all guards move model, which correspond to the 1-secure and $m$-secure variants, respectively. They showed that for perfect graphs, \textsc{Eternal Dominating Set} is efficiently solvable, leveraging the fact that $\gamma^\infty$ lies between the independence number and the clique covering number, and both are equal and efficiently computable in perfect graphs. Additionally, they established bounds for $\gamma^\infty$ and $\gamma_m^\infty$ in terms of other well-known graph parameters, such as the 2-domination number, clique covering number, and independence number. These bounds, in turn, motivate a closer look at the complexity distinctions between the related problems that are classical \textsc{Dominating Set}, the eternal version \textsc{Eternal Dominating Set}, and its all guard variant \textsc{$m$-Eternal Dominating Set}.

Subsequent work by Anderson et al. \cite{anderson2007maximum} introduced a game-theoretic perspective, modeling the eternal domination problem as a two-player game. Meanwhile, Klostermeyer and MacGillivray \cite{klostermeyer2009eternal} initiated a study of its computational complexity in trees. Among graph classes, chordal graphs are central to the algorithmic study of graph classes. They find applications in computational biology \cite{semple2003phylogenetics}, optimization \cite{vandenberghe2015chordal}, and sparse matrix computations \cite{george2012graph}. A key subclass of chordal graphs is the class of split graphs, where the vertex set can be partitioned into a clique and an independent set.

In \cite{macgillivrayeternal}, it was shown that for split graphs, $\gamma_m^\infty$ is either equal to $\gamma$ or $\gamma+1$. Moreover, they provided characterizations of graphs that attain eah of these values, and further proved that \textsc{$m$-Eternal Dominating Set} is \textsf{NP}-complete for Hamiltonian split graphs. Since this class is a subclass of split graphs, the result implies that the problem is \textsf{NP}-hard for both split graphs and, by extension, for chordal graphs. Interestingly, there exist two incomparable subclasses of split graphs: Hamiltonian split graphs and $K_{1,t}$-free split graphs. Because neither is a subset of the other, it is natural to investigate the complexity of $K_{1,t}$-free split graphs, and also to obtain a dichotomy based on the value of $t$. 

From the positive side, $\gamma_m^\infty$ is in \textsf{P} for several graph classes. These includs trees \cite{klostermeyer2009eternal}, proper interval graphs \cite{braga2015eternal}, interval graphs \cite{rinemberg2019eternal}, co-graphs, threshold graphs \cite{macgillivrayeternal} and cactus graphs \cite{blavzej2023computing}. These results show that while the problem may be hard on general graphs, there exist several natural graph classes where it is efficiently solvable. Since interval graphs form a subclass of chordal graphs, a natural step in narrowing the complexity gap between interval graphs (where the problem is in \textsf{P}) and chordal graphs (where it is \textsf{NP}-hard) is to study the complexity on undirected path graphs, which is an intermediate class between these two classes.

\begin{figure}
      \centering
      \includegraphics[height=4 cm, width=10 cm]{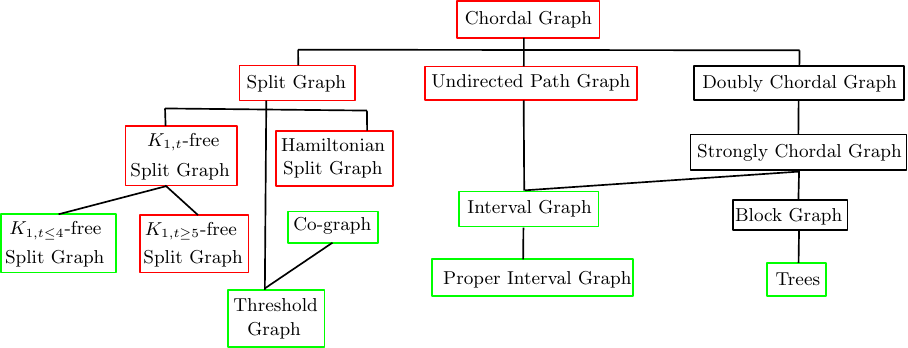}
      \caption{In the hierarchy diagram of chordal graph classes, an edge from one class to another indicates that the upper class is a superclass of the lower one. Each class is annotated with a colored box: red denotes that $m$-\textsc{Eternal Dominating Set} is \textsf{NP}-hard, green denotes polynomial-time solvability, and black indicates that the complexity remains unknown.}
      \label{fig:chordalgraphclasses}
  \end{figure}
  
For a comprehensive survey of results on eternal domination and related problems, we refer the reader to \cite{klostermeyer2016protecting}.
\vspace{-10pt}
 \subsection{Our work} 
 \vspace{-5pt}
  Motivated by these observations and complexity gaps, our work studies the computational complexity of the $m$-\textsc{Eternal Dominating Set} in subclasses of chordal graphs.
    The main contributions of this paper are summarized below.
\begin{itemize}
    \item 
    \textit We present a dichotomy result for split graphs. In \cite{macgillivrayeternal}, it is shown that computing $m$-\textsc{Eternal Dominating Set} is polynomial time solvable for any subclasses for which \textsc{Dominating Set} is already solvable in polynomial time. Furthermore, \cite{mohanapriya2023domination} proved that \textsc{Dominating Set} remains \textsf{NP}-complete for $K_{1,{t\ge 5}}$-free split graphs, while it becomes polynomial time solvable for $K_{1,{t\le 4}}$-free split graphs. In particular, for $K_{1,4}$-free split graphs, $m$-\textsc{Eternal Dominating Set} can be solved in $|V|^{\frac{5}{2}}$ time. Our algorithm improves upon this by achieving a faster $|V|^{\frac{3}{2}}$ running time on the same graph class. Motivated by this, we design a $|V|^{\frac{3}{2}}$ time algorithm for the $m$-\textsc{Eternal Dominating Set} on $K_{1,{t\le 4}}$-free split graphs, and we prove that the problem becomes \textsf{NP}-complete for $t \geq 5$. It is known that $m$-\textsc{Eternal Dominating Set} is in \textsf{NP} \cite{macgillivrayeternal}. The hardness result is established via a reduction from the \textsc{Exact 3-Cover} problem, which is known to be \textsf{NP}-complete \cite{karp2009reducibility}, while the polynomial time algorithm is based on analyzing the parameter $\Delta^I(G)$, which denotes the maximum number of neighbors in the independent set $I$ of a split graph $G=(C\cup I, E)$ for any vertex in the clique $C$.
    
\item We prove that $m$-\textsc{Eternal Dominating Set} is \textsf{NP}-hard for undirected path graphs, thereby closing the complexity gap between chordal graphs and interval graphs. The hardness is through a reduction from $3D$-matching problem by extending the ideas in \cite{booth1982dominating}, where \textsc{Dominating Set} was shown to be \textsf{NP}-hard on undirected path graphs.
    \item Further, we identify a graph class in which \textsc{Dominating Set} is \textsf{NP}-hard but the $m$-\textsc{Eternal Dominating Set} can be solved polynomial time. This graph class is constructed by taking any arbitrary graph and attaching $n$ disjoint copies of the path $P_5$ (a path on five vertices) to $n$ different vertices of the graph, where each $P_5$ is connected by joining its middle vertex to the corresponding vertex in the original graph. 
    \item Furthermore, we show the existence of two graph classes such that, in one, both \textsc{Dominating Set} and $m$-\textsc{Eternal Dominating Set} is solvable in polynomial time while \textsc{Eternal Dominating Set} is \textsf{NP}-hard, whereas in other, \textsc{Eternal Dominating Set} is solvable in polynomial time and both \textsc{Dominating Set} and $m$-\textsc{Eternal Dominating Set} are \textsf{NP}-hard. This highlights the complexity difference between the three related problems.  
\end{itemize}
\vspace{-10pt}
\section{Preliminaries}\label{sec:preliminaries}
\vspace{-5pt}
This section provides pertinent definitions and presents preliminary results that will be utilized in this paper.

Let $G=(V, E)$ be a finite, simple, undirected graph. The \textit{open neighborhood} of a vertex $v$ in $G$ is $N(v)=\{u \in V \mid uv \in E\}$ and the closed neighborhood is $N[v]=\{v\} \cup N(v).$ The \textit{degree} of a vertex $v$ is $|N(v)|$ and is denoted by $d_G(v).$ In case the graph is well understood, we denote it simply $d(v)$. If $d(v)=1$ then $v$ is called a \emph{pendant vertex} in $G$. 
The \textit{minimum degree} and \textit{maximum degree} of $G$ are denoted by $\delta(G)$ and $\Delta(G) $, respectively. For $D\subseteq V,~G[D]$ denotes the subgraph induced by $D.$ We use the notation $[k]$ for $\{1, 2, \cdots, k\}.$ We refer to \cite{west2001introduction} for remaining notations and terminology.

A graph is said to be a \emph{chordal graph} if every cycle of length at least 4 has a chord, i.e., an edge joining two non-consecutive vertices in the cycle. A \emph{split graph} $G=(V, E)$ is a subclass of the chordal graph in which the vertex set can be partitioned into two sets $V=C\cup I$ such that $C$ is a clique and $I$ is an independent set. A clique $C$ in a graph $G$ is said to be a \textit{maximal clique} if no other clique in $G$ properly contains $C$. We consider split graphs with a maximal clique (i.e., a split partition $V=C\cup I$ in which $C$ is maximal) for the rest of the paper. Note that $K_{1,t}$ is a split graph with $t+1$ vertices where $|C|=1$ and $|I|=t.$ A split graph $G$ is \emph{$K_{1,t}$-free} if $G$ does not contain $K_{1,t}$ as an induced subgraph. For a vertex $v\in C,$ $N_G^I(v)= N_G(v)\cap I$ and $d_G^I(v)=|N_G^I(v)|.$ For $S\subseteq C,$ $N_G^I(S)= \bigcup_{v\in S} N_G^I(v)$ and $d_G^I(S)=|N_G^I(S)|.$ For a split graph $G,$ $\Delta_G^I=\max\{d_G^I(v)\mid v\in C\}.$ A graph $G$ is said to be \emph{$l$-split graph} if $\Delta_G^I=l$, for some integer $l\ge 0$. A \emph{claw} is a complete bipartite graph $K_{1,3}$, consisting of a central vertex connected to three other vertices that are not connected. A split graph $G$ is a \emph{Hamiltonian Split Graph} if it contains a Hamiltonian cycle. For a graph $G = (V, E)$, a \textit{matching} $M \subseteq E$ in $G$ is a set of pairwise nonadjacent edges. A \emph{perfect matching} $M$ is a matching where each vertex $v\in V$ is incident to exactly one edge in $M$. The size of the maximum matching of $G$ is denoted as $\mu(G)$.
A graph $G=(V, E)$ is called an \emph{undirected path graph} if $G$ is the intersection graph of a family of paths of a tree. 

The following known results will be used throughout the paper. 

 \begin{theorem}\cite{gavril1974intersection}
     A graph $G=(V, E)$ is an undirected path graph if and only if there exists a tree $T$ (called as clique tree of $G$) with $V(T)=\mathcal{C}(G)$, where $\mathcal{C}(G)$ is the set of all maximal cliques of $G$, such that $T[\mathcal{C}_v(G)]$ is a path in $T$ for each $v\in V(G)$, where $\mathcal{C}_v(G)$ is the set of all maximal cliques of $G$ containing the vertex $v$.
 \end{theorem}
 
 \begin{theorem}\cite{renjith2020steiner}\label{union}
    Let $G=(C\cup I, E)$ be a connected split graph. Then $G$ is a claw-free graph if and only if either $\Delta_G^I \leq 1$ or $\Delta_G^I = 2$ then for every $u,v\in C$ such that $d_G^I(v)=2, N_G^I(u)\cap N_G^I(v)\neq \emptyset.$
\end{theorem}


\begin{lemma}\cite{renjith2020steiner}\label{delta}
    If a split graph is claw-free with $\Delta_G^I = 2,$ then $|I|\leq 3.$
\end{lemma}

\begin{observation}\cite{panda2021injective}\label{K1,3}
    Let $G=(C\cup I,E)$ be a $K_{1,3}$-free split graph. If for any pair of vertices $y_i, y_j \in I, N(y_i)\cap N(y_j)\neq \emptyset$,
then $N(y_i) \cup N(y_j) = C$.
\end{observation}

\begin{theorem}\cite{renjith2020steiner}\label{union_K_{1,4}}
    Let $G$ be a 3-split graph, then $G$ is $K_{1,4}$-free if and only if for every $v\in C$ such that $d_G^I(v)=3$, $N_G^I(u)\cap N_G^I(v) \neq \emptyset$, for each $u(\neq v)\in C$.
\end{theorem}

\begin{lemma}\cite{renjith2020steiner}\label{lemlsplit}
    Let $G$ be a $K_{1,4}$-free $3$-split graph. For any $v\in C$ such that $N_G^I(v)\neq \emptyset,$ the subgraph $H$, that is $G[V(G)\setminus N_G^I(v)]$, the induced subgraph on the vertex set $V(G)\setminus N_G^I(v)$, is a split graph for some $l, 0\leq l\leq 2.$ 
    
    Let $M$ be the labeled graph of $H$. Then the size of any maximum matching in $M, \mu(M)\leq 2.$
\end{lemma}

 \begin{observation}\cite{goddard2005eternal}\label{o1}
     Let $G$ be a graph. Then $\gamma(G) \le \gamma_m^\infty(G)\le \alpha(G),$ where $\alpha(G)$ and $\gamma(G)$ is the independence number and the domination number of $G$, respectively. 
 \end{observation}

  \begin{observation}\label{universal}
      Let $G$ be a non-complete graph. Then $G$ has a universal vertex if and only if $\gamma_m^\infty (G)=2$.
\end{observation}

\begin{observation}\label{maximal_independnet}
    Let $G=(V,E)$ be a graph and $S\subseteq V$ be such that $N(u)\cap N(v)=\emptyset$, for any two vertices $u,v$ of $S$. then $\gamma(G)\ge |S|$.
\end{observation}

\begin{observation}\label{ind_number_splitgraph}
    Let $G=(V=(C\cup I),E)$ be a split graph. Then the maximum independence number of $G$ is either $|I|$ or $|I|+1$. 
\end{observation}

\section{Complexity Difference}\label{sec:1}
In this section, we establish a complexity difference among three closely related graph problems: \textsc{Dominating Set}, \textsc{Eternal Dominating Set}, and $m$-\textsc{Eternal Dominating Set}. Specifically, we identify two distinct graph classes that exhibit contrasting computational complexities for these problems. In the first graph class, both \textsc{Dominating Set} and $m$-\textsc{Eternal Dominating Set} are \textsf{NP}-hard, whereas \textsc{Eternal Dominating Set} is solvable in polynomial time. In the second graph class, \textsc{Eternal Dominating Set} is \textsf{NP}-hard, while \textsc{Dominating Set} and $m$-\textsc{Eternal Dominating Set} admit polynomial-time solutions. 
Now, let us define a graph class $GP_3$.
\begin{definition}($GP_3$-graph) \label{def:G^*}
 Let $G=(V, E)$ be any graph. A graph $G'=(V', E')$ is called a $GP_3$-graph if it can be constructed from $G$ by attaching $n$ distinct paths of length 3, $P_3^i=\{v_i^0,v_i^1,v_i^2\}, i\in [n]$ to $G$ such that $v_i\in V$ is adjacent to $v_i^1$ (the middle vertex of $P_3^i$), for each vertex $v_i$, $i\in [n]$. 
 
 $GP_3$-graph class consists of all the $GP_3$-graphs.
 \end{definition}

 Note that for a $GP_3$-graph $G'=(V',E')$, constructed from the graph $G=(V,E)$, we have $|V'|=4|V|$ and $|E'|=3|V|+|E|$.

We first show that \textsc{Dominating Set} and $m$-\textsc{Eternal Dominating Set} can be computed in linear time on $GP_3$ graphs.
\begin{lemma}\label{GP_3gamma}
    For a $GP_3$ graph $G'$, $\gamma(G')=\frac{|V'|}{4}$ and $\gamma_m^\infty(G')=\frac{|V'|}{2}$.
\end{lemma}
\begin{proof}
 First, we prove that $\gamma(G')=\frac{|V'|}{4}$. Let $G$ be any graph and $G'$ be the corresponding $GP_3$ graph. Let $D'$ be any dominating set of $G'$. As $D'$ is a dominating set, it follows that $|\{v_i^0,v_i^1,v_i^2\}\cap D'|\ge 1$, for all $i\in [n]$. Thus, $|D'|\ge |V|=\frac{|V'|}{4}$. 
    
    Conversely, let $D'=\{v_i^0\mid 1\le i\le n\}$. Since $\displaystyle\cup_{i\in [n]} N[v_i^0]=V'$, $D'\subseteq V'$ is a dominating set of $G'$. Therefore, $\gamma(G')\le |V|=\frac{|V'|}{4}$.\\
    
 Now, we prove that $\gamma_m^\infty(G')=\frac{|V'|}{2}$. 
     Let $G$ be any graph and $G'$ be the corresponding $GP_3$ graph. Let $D'$ be an $m$-eternal dominating set of $G'$. As $D'$ is a dominating set, it follows that $|\{v_i^0,v_i^1,v_i^2\}\cap D'|\ge 1$, for all $i\in [n]$. If for each $P_3^i, i\in [n]$, $|D'\cap P_3^i|\ge 2$ then $|D'|\ge 2|V|=\frac{|V'|}{2}$. Otherwise, let there exists a $P_3^j$ for some $j\in [n]$ such that $|D'\cap P_3^j|< 2$. Also, we have $|D'\cap P_3^j|\ge 1$ because $|\{v_i^0,v_i^1,v_i^2\}\cap D|\ge 1$, for all $i\in [n]$. This implies that $|D'\cap P_5^j|= 1$. So, clearly $v_j^1\in D'$. If $v_j\notin D'$, then the attack on $v_j^2$ can not be defended, a contradiction to $D'$ being an $m$-eternal dominating set. Thus, $v_j\in D'$. Now, the attack on $v_j^2$ can be defended by $v_j^1\rightarrow v_j^2, v_j\rightarrow v_j^1$, implying $|D'\cap P_3^j|\ge 2$. Thus, $|D'\cap \{v_j^0,v_j^1,v_j^2\}|\ge 2$, implying $|D'\cap P_3^j|\ge 2$. Hence, $|D'|\ge 2|V|=\frac{|V'|}{2}$. 
    
    Conversely, let $D=\{v_i^1,v_i^2\mid 1\le i\le n\}$. Since $\{v_i,v_i^1,v_i^2,v_i^0\mid i\in [n]\}$ forms an induced star of $G'$ having $v_i^1$ as the central vertex, by maintaining $v_i^1$ as invariant for each $i\in [n]$, any attack on leaves of the induced star can be defended. Thus, $D'$ is an $m$-eternal dominating set of $G'$. Therefore, $\gamma_m^\infty(G')\le 2|V|=\frac{|V'|}{2}$. 
\end{proof}

Now, we provide the hardness of \textsc{Eternal Dominating Set} on $GP_3$ graphs.

 \begin{lemma}\label{lem_NP-eternal}
    Let $G = (V, E)$ be a graph of order $n$ and let $G' = (V',E')$ be the corresponding $GP_3$-graph. Then for any integer $k \le n$, the graph $G$ has an eternal dominating set of cardinality at most $k$ if and only if $G'$ has an eternal dominating set of cardinality at most $k + 2n$.
\end{lemma}
\begin{proof}
   Let $D\subseteq V$ be an eternal dominating set of $G$ of cardinality at most $k$. Define $D'=D\cup \{v_i^0,v_i^2\mid i\in [n]\}$. Clearly, $|D'|\le k+2n$. We now show that $D'$ is an eternal dominating set of $G'$. First, observe that $D'$ is a dominating set of $G'$: since $D$ dominates $G$ and for each $i\in [n]$, $N[\{v_i^0,v_i^2\}]=\{v_i^0,v_i^1,v_i^2\}$. Now, consider an attack on any vertex $v\in V'$. If $v\in V$, the attack can be defended using the same strategy as in $G$. If $v\in \{v_i^0,v_i^1,v_i^2\}$ for some $i\in [n]$, observe that the set $\{v_i^0,v_i^1,v_i^2\}$ forms an induced star in $G'$, centered at $v_i^1$. By maintaining $v_i^1$ as invariant, any attack on leaves $v_i^0$ or $v_i^2$ can be defended. Moreover, it is easy to see that all the new configuration of guards remains a dominating set of $G'$. Hence, $D'$ is an eternal dominating set of $G'$. 

   Conversely, let $D'\subseteq V'$ be an eternal dominating set of $G'$ with $|D'|\le k+2n$. Since $D'$ is a dominating set of $G'$, we have for each $i\in [n]$, $|D'\cap \{v_i^0,v_i^1,v_i^2\}|\ge 1$.
  Moreover, $|D'\cap \{v_i^0,v_i^1,v_i^2\}|\ge 2$ as $D'$ is an eternal dominating set of $G'$. Suppose there exists some $i\in [n]$ such that in all configurations, $|D'\cap \{v_i^0,v_i^1,v_i^2\}|= 3$. Then it is easy to see that $D^*=D\setminus \{v_i^2\}$ is also an eternal dominating set of $G'$ of smaller cardinality. Now consider the case where, in some configuration, $|D'\cap \{v_i^0,v_i^1,v_i^2\}|=3$ for some $i\in [n]$. Since the model allows only one guard to move per attack, it must be the case that the attack occurred at $v_i^1$ and was defended by $v_i$. We can modify the defense strategy such that $|D'\cap \{v_i^0, v_i^1, v_i^2\}| = 2$, and still other attacks can be defended; then we are done. Instead, $v_i^0$ defends the attack at $v_i^1$ and the guard at vertex $v_i$ stays where it is. Observe that any attack on $v_i$ can be defended by itself, and for all other vertices, it can be defended using the previous strategy. Hence, for each $i\in [n]$, $|D'\cap \{v_i^0,v_i^1,v_i^2\}|=2$. Thus, in all configurations, $|D'\cap \{v_i^0,v_i^1,v_i^2\}|= 2$ for all $i\in [n]$. Then define $D=D'\setminus (D'\cap \{v_i^0,v_i^1,v_i^2\})$ implies $|D|\le k$. Now, it remains to prove that $D$ is an eternal dominating set of $G$. We show this by mimicking the moves of $D'$ in $G$. If the attack occurs on $v_i$, and the guard had moved in the original strategy for $D'$ from $v_i$ to $v_i^1$, then we do not need to take any action. For any other attack on some $v_i$, the guard defending the attack for $D'$ can defend the attack in the modified strategy as well.  
\end{proof}

Since the \textsc{Eternal Dominating Set} is known to be \textsf{NP}-hard for general graphs \cite{virgile2024mobile}, the \textsf{NP}-hardness of the \textsc{Eternal Dominating Set} on $GP_3$-graphs follows directly from Lemma~\ref{lem_NP-eternal}. 

It is known that \textsc{Eternal Dominating Set} is polynomial time solvable for perfect graphs \cite{goddard2005eternal}. Thus, as a corollary, we have:
  \begin{corollary}
    \textsc{Eternal Dominating Set} is polynomial time solvable for chordal graphs. 
 \end{corollary}
However, for the chordal graphs, the other two problems are \textsf{NP}-hard. That is, \textsc{Dominating Set} is \textsf{NP}-complete for chordal graphs \cite{booth1982dominating}. Moreover, it is also known that $m$-\textsc{Eternal Dominating Set} is \textsf{NP}-complete for Hamiltonian split graphs \cite{macgillivrayeternal}. Since Hamiltonian split graphs form a subclass of chordal graphs, we obtain the following corollary.
 \begin{corollary}
     $m$-\textsc{Eternal Dominating Set} is \textsf{NP}-hard for chordal graphs.
 \end{corollary}

 \subsection{\textsc{Dominating Set} vs. $m$-\textsc{Eternal Dominating Set}}
\label{sec:Complexity Difference}
\vspace{-5pt}
In this subsection, we identify a graph class, where \textsc{Dominating Set} is \textsf{NP}-hard while $m$-\textsc{Eternal Dominating Set} is solvable in polynomial time.

Let $P_5$ denote a path graph on five vertices. For $1 \le i \le n$, let $\{P_5^i\mid 1 \le i \le n\}$ be a collection of $n$ distinct 5-vertex paths, where $P_5^i=v_i^0v_i^1v_i^2v_i^3v_i^4$. Now, we formally define the graph class $GP_5$-graph as follows.

\begin{definition}($GP_5$-graph)
 Let $G=(V, E)$ be any graph. A graph $G_P=(V_P,E_P)$ is called a $GP_5$-graph if it can be constructed from $G$ by attaching $v_i^2$ (the middle vertex of the path $P_5^i$), to the corresponding vertex $v_i$ in $V$, $i\in [n]$. 
 
 $GP_5$-graph class consists of all the $GP_5$-graphs.
\end{definition}
Note that for a $GP_5$-graph $G_P=(V_P,E_P)$, constructed from the graph $G=(V,E)$, we have $|V_P|=6|V|$ and $|E_P|=5|V|+|E|$.

We first show that $m$-\textsc{Eternal Dominating Set} can be computed in linear time on $GP_5$ graphs.

\begin{lemma}\label{GP_5}
    For a $GP_5$ graph $G_P,$ $\gamma_m^\infty(G_P)=\frac{|V_P|}{2}$. Hence,  $m$-\textsc{Eternal Dominating Set} can be computed in linear time for $GP_5$ graphs.
\end{lemma}
\begin{proof}
    Let $G$ be any graph and $G_P$ be the corresponding $GP_5$ graph. Let $D$ be an $m$-eternal dominating set of $G_P$. As $D$ is a dominating set, it follows that $|\{v_i^0,v_i^1\}\cap D|\ge 1$ and $|\{v_i^3,v_i^4\}\cap D|\ge 1$, for all $i\in [n]$. If for each $P_5^i, i\in [n]$, $|D\cap P_5^i|\ge 3$ then $|D|\ge 3|V|=\frac{|V_P|}{2}$. Otherwise, let there exists a $P_5^j$ for some $j\in [n]$ such that $|D\cap P_5^j|\le 2$. Also, we have $|D\cap P_5^j|\ge 2$ because $|\{v_i^0,v_i^1\}\cap D|\ge 1$ and $|\{v_i^3,v_i^4\}\cap D|\ge 1$, for all $i\in [n]$. This implies that $|D\cap P_5^j|= 2$. Now, if $|D\cap \{v_j^1,v_j^3\}|=0$, then attack on $v_j^2$ can only be defended by $v_j$, implying $|D\cap P_5^j|\ge 3$; else let $|D\cap \{v_j^1,v_j^3\}|\le 1$. Without loss of generality, let $D\cap \{v_j^1,v_j^3\}=v_j^1$. Observe that the attack on $v_j^2$ can not be defended because if the only neighbor $v_j^1\in D$ moves to $v_j^2$, the vertex $v_j^0$ remains undominated. Thus, $|D\cap \{v_j,v_j^2\}|\ge 1$, implying $|D\cap P_5^j|\ge 3$. Hence, $|D|\ge 3|V|=\frac{|V_P|}{2}$. 
    
    Conversely, let $\hat{D}=\{v_i^1,v_i^2,v_i^3\mid 1\le i\le n\}$. Since $\{v_iv_i^2,v_i^0v_i^1,v_i^3v_i^4\mid i\in [n]\}$ forms a perfect matching of $G_P$, $v_i^1,v_i^2,v_i^3$ defends the attack on vertices of the respective edges $v_i^0v_i^1,v_iv_i^2,v_i^3v_i^4$. Thus, $\hat{D}$ is an $m$-eternal dominating set of $G_P$. Therefore, $\gamma_m^\infty(G_P)\le 3|V|=\frac{|V_P|}{2}$. 
\end{proof}

Next, we show that \textsc{Dominating Set} is \textsf{NP}-hard for $GP_5$ graphs. To do this, we prove that the \textsc{Dominating Set} for a general graph $G$ is efficiently solvable if and only if it is efficiently solvable for the corresponding $GP_5$-graph.

\begin{lemma}\label{lem_GP_5}
    Let $G = (V, E)$ be a graph of order $n$ and let $G_P = (V_P,E_P)$ be the corresponding $GP_5$-graph. Then for any integer $k \le n$, the graph $G$ has a dominating set of cardinality at most $k$ if and only if $G_P$ has a dominating set of cardinality at most $k + 2n$.
\end{lemma}
\begin{proof}
    Let $D\subseteq V$ be a dominating set of $G$ such that $|D|\le k$. Define $D'=D\cup \{v_i^1,v_i^3\mid i\in [n]\}$. It is easy to verify that $D'$ is a dominating set of $G_P$ and $|D'|\le k+2n$.

    On the other hand, let $D'\subseteq V_P$ be a dominating set of $G_P$ such that $|D'|\le k+2n$. Note that, for each $i\in [n]$, $|D'\cap \{v_i^0,v_i^1\}|\ge 1$ and $|D'\cap \{v_i^3,v_i^4\}|\ge 1$. If $v_i^2\in D'$ then let $D^*=D'\setminus \{v_i^2\} \cup \{v_i\}$. Now, we assume that $v_i^2\notin D^*$ for any $i\in [n]$. Define $D=D^*\cap V$. Then $|D|\le k$, as $|D\cap \{v_i^0, v_i^1, v_i^2, v_i^3, v_i^4\}|\ge 2$. Since, $v_i^2\notin D^*$ for any $i\in [n]$,  $D$ is a dominating set of $G$.
\end{proof}

Since the \textsc{Dominating Set} is known to be \textsf{NP}-hard for general graphs \cite{bertossi1984dominating}, the \textsf{NP}-hardness of the \textsc{Dominating Set} on $GP_5$-graphs follows directly from Lemma~\ref{lem_GP_5}.

\begin{definition}($GP_2$-graph)
 Let $G=(V, E)$ be any graph. A graph $G'=(V',E')$ is called a $GP_2$-graph if it can be constructed from $G$ by attaching $n$ distinct pendant edges $\{v_i^1v_i^2\mid i\in [n]\}$ such that $v_i$ is adjacent to $v_i^1$ for each $i\in [n]$. 
 
 $GP_2$-graph class consists of all the $GP_2$-graphs.
 \end{definition}

 \begin{observation}
      For a $GP_2$-graph $G'$, $\gamma(G')=\frac{|V'|}{3}$.
 \end{observation}
    
Now, we conjecture the following.
 \begin{conjecture}\label{lem_GP_2}
    Let $G = (V, E)$ be a graph of order $n$ and let $G' = (V',E')$ be the corresponding $GP_2$-graph. Then for any integer $k \le n$, the graph $G$ has an $m$-eternal dominating set of cardinality at most $k$ if and only if $G'$ has an $m$-eternal dominating set of cardinality at most $k + n$.
\end{conjecture}

The above Conjecture \ref{lem_GP_2} implies that $m$-\textsc{Eternal Dominating Set} is \textsf{NP}-hard for $GP_2$-graphs.

The Table \ref{tab:complexity_diff} summarizes the results obtained in this section.

\begin{table}[h!]
    \centering
    \begin{tabular}{|c|c|c|c|}
    \hline
         & Problem $1$ & Problem $2$ & Problem $3$ \\
         \hline
       Problem $1$  & $-$ & Chordal graph & $GP_5$-graph\\
         \hline
       Problem $2$  & $GP_3$-graph & $-$ & $GP_3$-graph\\
         \hline
         Problem $3$  & $?$ & Chordal graph & $-$\\
         \hline
    \end{tabular}
    \caption{The cell $(i,j)$ denote a graph class where the $i^{th}$ problem is harder and $j^{th}$ problem is efficiently solvable. Here, Problem $1$ refers to \textsc{Dominating Set}, Problem $2$ refers to \textsc{Eternal Dominating Set}, and Problem $3$ refers to \textsc{Eternal Dominating Set}. `$?$' denotes whether there exists a graph class (where $i^{th}$ problem is harder and $j^{th}$ problem is efficiently solvable) is still unknown.}
    \label{tab:complexity_diff}
\end{table}
\section{Complexity Dichotomy for Split Graphs}
\label{sec:$K_{1,t}$-free-SG}
In this section, we present a dichotomy result for split graphs. Note that every split graph is $K_{1,t}$-free for some $t \ge 1$. Since the $m$-eternal domination number of a disconnected graph equals the sum over its connected components, we focus on connected graphs. It is easy to see that $\gamma_m^\infty(K_{1,n}) = 2$ for $n \ge 1$, and $\gamma_m^\infty(G) = 1$ for a complete graph $G$. Thus, we consider only connected graphs that are neither complete nor stars for the rest of this section.  

\subsection{$K_{1,3}$-free Split Graphs}
\vspace{-5 pt}
We now present a characterization of the $m$-eternal domination number for connected $K_{1,3}$-free split graphs. This result also leads to a linear-time algorithm for computing this parameter.
\begin{theorem}\label{alg_K_{1,3}}
Let $G = (C \cup I, E)$ be a connected $K_{1,3}$-free split graph, where $C$ is a clique and $I$ is an independent set. Then, the $m$-eternal domination number of $G$, $\gamma_m^\infty(G)$, is given by
$\gamma_m^\infty(G) = 
\begin{cases}
|I| + 1, & \text{if } \cup_{v \in I} N(v) \neq C; \\
|I|, & \text{if } \cup_{v \in I} N(v) = C.
\end{cases}$
\end{theorem}
\begin{proof}
\begin{itemize}
    \item {\emph{Case 1}}: ($\cup_{v \in I} N(v) \neq C$)
    Since $G$ is a $K_{1,3}$-free split graph and $\cup_{v \in I} N(v) \neq C$, we have $\Delta_G^I=1$. Thus, for any pair of vertices $v_1,v_2\in I$, $N(v_1)\cap N(v_2)=\emptyset$. This implies that $\gamma(G)\ge |I|$ (from Observation \ref{maximal_independnet}) and $\alpha(G)\le |I|+1$ (from Observation \ref{ind_number_splitgraph}). Thus, we have $|I|\le \gamma_m^\infty(G)\le |I|+1$ (by Observation \ref{o1}). To show $\gamma_m^\infty(G)>|I|$, i.e., with $|I|$ number of guards, it is not always possible to defend all the attacks on the graph $G$. Assume for contradiction that there exists a minimal $m$-eternal dominating set $D$ such that $|D|=|I|$. 
    Since, $\cup_{v \in I} N(v) \neq C$, there exist a vertex (say $x$) in $C\setminus N(I)$ and as $D$ is also a dominating set, $D\cap N(I)\neq \emptyset$. Then an attack at $x\in C\setminus N(I)$ can not be defended by moving any guard from $C$ without leaving some vertex in $I$ undominated. Hence, $\gamma_m^\infty(G)=|I|+1$.

    \item {\emph{Case 2}}: ($\cup_{v \in I} N(v) = C$)
    Since $G$ is a $K_{1,3}$-free split graph and $\cup_{v \in I} N(v) = C$, we have the following two subcases. 
    \begin{itemize}
        \item {\emph{Subcase 1}}: ($\Delta_G^I=1$) Since for any pair of vertices $v_1,v_2\in I$, $N(v_1)\cap N(v_2)=\emptyset$, it implies that $\gamma(G)\ge |I|$ (from Observation \ref{maximal_independnet}). Thus, we have $\gamma_m^\infty(G)\ge |I|$ (from Observation \ref{o1}). Moreover, because $\cup_{v_i\in I} N(v_i) = C$, we have $\gamma_m^\infty(G)\le \alpha(G)\le |I|$ (from Observation \ref{ind_number_splitgraph}). Hence, $\gamma_m^\infty(G)= |I|$. 
        
        \item {\emph{Subcase 2}}: ($\Delta_G^I=2$) Since $G$ is not a complete graph, we have $|I|\ge 2$ and from Lemma \ref{delta}, we have $|I|\le 3$. If $|I|=2$, there exist a universal vertex in $C$ as $\Delta_G^I=2$. In this case, $\gamma_m^\infty(G)=2=|I|$ (by Observation \ref{universal}). Now consider the case $|I|=3$. Since $\Delta_G^I=2$, no universal vertex exists. So, by  Observation \ref{universal}, $\gamma_m^\infty(G)\ge 3$. We now show that $\gamma_m^\infty(G)\le 3$. Observe that $G[N[v_i]]$ is a complete graph for each $v_i\in I$. Thus, maintaining one vertex in each $G[N[v_i]]$, any attack on the vertices of $G$ can be defended. Hence, $\gamma_m^\infty(G)=3$.
    \end{itemize}
\end{itemize}
\end{proof}

\emph{Time Complexity:} Assume the graph $G=(C\cup I, E)$ is given as an adjacency list representation. Constructing the union $\cup_{v \in I} N(v)$ and verifying whether this union equals the clique $C$ can be done simultaneously by first initializing an array of size $|C|$ indexed by vertices of $C$ (takes $O(|V|)$ time) and then for each vertex of $I$, scan through its adjacency list and mark its neighbors in the array (takes $O(|V|+|E|)$ time), then scan the array again to check whether all vertices are marked or not (takes $O(|V|)$ time). Therefore, the overall computation can be performed in $O(|V| + |E|)$ time. Thus, the m-eternal domination number of a connected $K_{1,3}$-free split graph can be computed in linear time.

\subsection{$K_{1,4}$-free Split Graphs}
\label{sec:$K_{1,4}$-free-SG}
We now present a characterization of the $m$-eternal domination number for connected $K_{1,4}$-free split graphs. This class extends the result obtained for $K_{1,3}$-free split graphs in the previous subsection. Furthermore, our characterization leads directly to a polynomial time algorithm for computing this parameter.

\begin{observation}\label{4to3split}
Every $K_{1,4}$-free $1$-split graph is also a $K_{1,3}$-free $1$-split graph.
\end{observation}

\begin{definition}\cite{renjith2020steiner}\label{labeled graph} (labeled graph)
For an $l$-split graph $H=(\hat{C}\cup \hat{I},\hat{E}), 0\leq l\leq 2$, the labeled graph $M=(V_M,E_M)$ is defined such that $V_M=\hat{I}$ and $E_M=\{uv\mid u,v\in \hat{I}, N_H(u)\cap N_H(v) \neq \emptyset\}$. Each edge $uv$ is labeled by a vertex $w_{uv} \in N_H(u)\cap N_H(v)$, which is referred to as its labeled clique vertex.     
\end{definition} 

\begin{theorem}\label{alg_K_{1,4}-2SG}
Let $H=(\hat{C}\cup \hat{I},\hat{E})$ be a $K_{1,4}$-free $2$-split graph, and let $M$ be its labeled graph. Denote $\mathcal{M}$, a maximum matching of $M$, and let $L$ be the set of vertices in $\hat{C}$ corresponding to the edges of $\mathcal{M}$. Construct the reduced graph $H'=(\hat{C}\cup I',E')$, where $I'=\hat{I}\setminus N_{H}^{\hat{I}}(L), 
E'=\hat{E}\setminus \{l_v i_v \mid l_v\in L,\, i_v\in N_H^{\hat{I}}(L)\}$. Then the $m$-eternal domination number of $H$ is given by
$
\gamma_m^\infty(H)= |L|+|I'|+1$.
\end{theorem}
\begin{proof}
Since $H=(\hat{C}\cup \hat{I}, \hat{E})$ is a $K_{1,4}$-free $2$-split graph, $\Delta_{H}^{\hat{I}}=2$. After forming a labeled graph $M$ followed by deleting $N_H^{\hat{I}}(l_v)$ of each labeled clique vertex $l_v$ of $L$, corresponding to all edges in $\mathcal{M}$ (a maximum matching of the labeled graph), yields the reduced graph $H'=(\hat{C}\cup I', E')$, where $I'=\hat{I}\setminus N_H^{\hat{I}}(L)$ and $E'=\hat{E}\setminus \{l_vi_v\mid l_v\in L, i_v\in N_H^{\hat{I}}(l_v)\}$. Observe that, the resulting graph $H'$ is a $K_{1,3}$-free $1$-split graph as $\Delta_{H'}^{I'}\le 1$. Next, we analyze two cases based on $\Delta_{H'}^{I'}$.
\begin{itemize}
    \item {\emph{Case 1}}: ($\Delta_{H'}^{I'}=0$) Let us denote the vertices of the matching edge $m\in \mathcal{M}$ by $v^1_m$ and $v^2_m$ and their corresponding labeled clique vertex by $l_m$. Observe that $l_{m_i}\neq l_{m_j}$ for any two matching edges $m_i,m_j\in \mathcal{M}$. Let $S\subseteq N_H^{\hat{I}}(L)$ consists of exactly one neighbor of $l_m$ for each $l_m\in L$. Clearly, $|S|=|L|$ and $N(x)\cap N(y)=\emptyset$ for any vertices $x,y\in S$. Thus, any dominating set of $H$ must have at least $|L|$ vertices, i.e., $\gamma(H)\ge |L|$ (from Observation \ref{maximal_independnet}). So, $\gamma_m^\infty(H)\ge |L|$ follows from Observation \ref{o1}. Now, we show that $\gamma_m^\infty(H)>|L|$, i.e., with $|L|$ number of guards, it is not always possible to defend against attacks on the vertices of $H$. Let $D$ be a minimal $m$-eternal dominating set of $H$ of cardinality $L$. In fact, $D=L$, as each vertex in $L$ has two neighbors in $\hat{I}$. An attack on any vertex of the matching edge $m_i\in \mathcal{M}$ (without loss of generality, let $v^1_{m_i}$), can not be defended as moving the only adjacent guard $l_{m_i}$ to $v^1_{m_i}$, leaves the vertex $v^2_{m_i}$ undominated. Hence, $\gamma_m^\infty(H)\ge |L|+1$. Now, let a set $D=L\cup \{i_v\}$, where $i_v\in \hat{I}$ (note that $i_v=v_m^j, j\in [2]$ for some $m\in \mathcal{M}$ as $N_H^{\hat{I}}=\hat{I}$). Next, to show $\gamma_m^\infty(H)\le |L|+1$, we show that $D$ is an $m$-eternal dominating set of $H$. See that $D$ is a dominating set of $H$ since $L\subseteq D$. Define $D_1=D$ and $D_2=L\cup \{c^*\}$, where $c^*\in \hat{C}\setminus L$. By maintaining the guards at $L$ as an invariant in every configuration, we defend the attacks on the vertices of $H$ by using the strategy described in Table \ref{tab:K_1,4_case1.2}. Hence, $\gamma_m^\infty(H)=1+|L|$. Since, $\Delta_{H'}^{I'}=0$, we have $|I'|=0$. Thus, $\gamma_m^\infty(H)=|L|+|I'|+1$. 
    
        \begin{table}[h!]
     \centering
    \renewcommand{\arraystretch}{1.3}
    \begin{tabular}{|c|>{\centering\arraybackslash}p{5 cm}|>{\centering\arraybackslash}p{4cm}|}
    \hline
    \multirow{2}{3 cm}{Vertex on which attack happened} & 
    \multicolumn{2}{c|}{Type of guard configuration at the time of attack} \\
    \cline{2-3}
    & $D_1$ & $D_2$ \\
         \hline
         $\{c_v, c_v\in \hat{C}\setminus L\}$& $D_2 ~\{l_m\rightarrow c_v, i_v\rightarrow l_m\}$ & $D_2 \{c^*\rightarrow c_v\} $\\
         \hline
        $\{i^*, i^*\in \hat{I}\}$ & $D_1 ~\{l_{m^*}\rightarrow i^*, l_m\rightarrow l_{m^*}, i_v\rightarrow l_m\}$ & $D_1~ \{l_{m^*}\rightarrow i^*, c^*\rightarrow l_{m^*}\}$\\
         \hline
    \end{tabular}
    \caption{\centering Defending strategies against vertex attacks, illustrated for case 1 in the proof of Theorem \ref{alg_K_{1,4}-2SG}. Here $i_v\in \{v_m^1,v_m^2\}$ and $i^*\in \{v_{m^*}^1,v_{m^*}^2\}$.}
    \label{tab:K_1,4_case1.2}
\end{table}

    \item {\emph{Case 2}}: ($\Delta_{H'}^{I'}=1$) Let $A=I'\cup X^*$, where $X^*$ is a maximal independent subset of $X=\{u\in N_H^{\hat{I}}(v)\mid v\in L$ and $N_H(u)\cap N_H(I')=\emptyset\}$. Now, we claim that $|X^*|=|L|$. Since $\mathcal{M}$ is a maximum matching of $M$, there exists at least one neighbor of $l_m\in \mathcal{M}$ (without loss of generality take $v_m^1$, say $N(l_m)=\{v_m^1, v_m^2\}$) such that $N_H(v_m^1)\cap N_H(I')=\emptyset$. Thus, $|X^*|\ge |L|$. Also, $|X^*|\le |L|$ holds because $X^*$ is an independent set implies $\{v_m^1, v_m^2\}\nsubseteq X^*$. Thus, $|A|=|L|+|I'|$. Now, observe that no two vertices of $A$ have a common neighbor. Thus, any dominating set of $H$ must require at least $|A|$ vertices, i.e., $\gamma(H)\ge |A|$ (from Observation \ref{maximal_independnet}). This follows that $\gamma_m^\infty(H)\ge |A|$ (by Observation \ref{o1}). Now, we show $\gamma_m^\infty(H)>|A|$, i.e., with $|A|$ number of guards, it is not always possible to defend the attacks on the graph $H$. Let $D$ be a minimal $m$-eternal dominating set of $H$ such that $|D|=|A|$. Then there are two cases: either $N(v_m^2)\cap N(I')=\emptyset$ or $N(v_m^2)\cap N(I')\neq\emptyset$. First, let $N(v_m^1)\cap N(I')=\emptyset$. Clearly, $l_m\in D$ as there are $|A|$ guards. Thus, an attack on $v_m^1$ (resp. $v_m^2$) makes $v_m^2$ (resp. $v_m^1$) undominated. Now, consider the case $N(v_m^2)\cap N(I')\neq\emptyset$. Let $i_m\in I'$ such that $N(v_m^2)\cap N(i_m)= \{c_m\}$. Clearly, $\{l_m,c_m\}\cap D\neq \emptyset$ as there are $|A|$ guards. Thus, attack on $v_m^2$ makes $v_m^1 ($ resp. $i_m)$ undominated if $l_m ($ resp. $c_m)\in D$ defends. Hence, $\gamma_m^\infty(H)\ge |A|+1= 1+|L|+|I'|$. Next, to show $\gamma_m^\infty(H)\le 1+|L|+|I'|$, we take a set $D$ of cardinality $1+|L|+|I'|$ and show that $D$ is an $m$-eternal dominating set of $H$. Observe that $H'$ is a $K_{1,4}$-free $1$-split graph. So, $d_{H'}^{I'}(v_i)\le 1$ for each $v_i\in C$. This follows that for any pair $v_i,v_j\in I'$, $N(v_i)\cap N(v_j)=\emptyset$. Let $C_i=N(v_i)$, for each $v_i\in I'$. Now, let $S=\{c_i\in C_i\mid c_i$ is a representative of $C_i\}$ and $D=S\cup L \cup \{i^*\}$ (note that $|S|=|I'|$), where $i^*\in \hat{I}$. Now, we show that $D$ is an $m$-eternal dominating set of $H$. Observe that $D$ is a dominating set of $H$ as $S\cup L \subseteq D$. Define $D_1=D$ and $D_2=D\setminus \{i^*\}\cup \{c^*\}$, where $c^*\in \hat{C}\setminus (S\cup L)$. Clearly, $|D_1|=|D_2|$. By maintaining the set $S\cup L$ as invariant, we defend the attacks on any vertex of $H$ by using the strategy described in Table \ref{tab:K_1,4_case2}. Therefore, $\gamma_m^\infty(H)= |L|+|I'|+1$. 
\end{itemize}
\begin{table}[h!]
     \centering
    \renewcommand{\arraystretch}{1.3}
    \begin{tabular}{|c|>{\centering\arraybackslash}p{4.5cm}|>{\centering\arraybackslash}p{4.5cm}|}
    \hline
    \multirow{2}{3 cm}{Vertex on which attack happened} & 
    \multicolumn{2}{c|}{Type of guard configuration at the time of attack} \\
    \cline{2-3}
    & $D_1$ & $D_2$ \\
         \hline
         $\{c_v, c_v\in C\setminus (S\cup L) \}$&  $D_2  \{v^*\rightarrow c_v, i^*\rightarrow v^*\}$  & $D_2 ~\{c^*\rightarrow c_v\}$\\
         \hline
        $\{i_v, i_v\in I\}$ & $D_1 ~\{v_i\rightarrow i_v, v^*\rightarrow v_i, i^*\rightarrow v^*\}$ & $D_1 \{v_i\rightarrow i_v, c^*\rightarrow v_i\}$\\
         \hline
    \end{tabular}
    \caption{\centering Defending strategies against vertex attacks, illustrated for case 2 in the proof of Theorem \ref{alg_K_{1,4}-2SG}. Here $v^*\in N(i^*)\cap D_2$, $v_i\in N(i_v)\cap D_1$.}
    \label{tab:K_1,4_case2}
\end{table}
\end{proof}

\emph{Time Complexity:} Let the graph $H=(\hat{C}\cup \hat{I}, \hat{E})$ be given in adjacency list representation. Constructing the labeled graph $M$ from $H$ requires $O(|V|)$ time. Since $H$ is a $2$-split graph, each vertex $c\in \hat{C}$ has at most two neighbors in $\hat{I}$; thus, by iterating over the vertices of $\hat{C}$, we can generate all edges of $M$ together with their corresponding in $O(|V|)$ time. Finding a maximum matching $\mathcal{M}$ of $M$ can be done in $O(|V|^{\frac{3}{2}})$ time using a standard matching algorithm \cite{micali1980v}. Collecting all labels $L\subseteq \hat{C}$ of the matching edges of $\mathcal{M}$ and constructing the graph $H'=(\hat{C}\cup I',E')$ by removing $N_H^{\hat{I}}(L)$ and their incident edges takes $O(|V|)$ time. Therefore, the overall running time of the algorithm derived from Theorem~\ref{alg_K_{1,4}-2SG} is $O(|V|^{\frac{3}{2}})$.

\begin{definition}(Type I and Type II)
    Let $G=(X\cup Y, E)$ be a bipartite graph such that $|X|=3$ and there exists a weight function $w: Y\rightarrow \mathbb{N}$. Then
    \begin{itemize}
        \item $G$ is said to be of Type I if there exists a matching in $G$ that saturates all vertices of $X$ such that no two saturated vertices in $Y$ have the same weight.
        \item $G$ is said to be of Type II if it contains two disjoint induced subgraphs: $K_{1,2}=\{ac, bc\}$ and $K_2=\{de\}$, such that $a,b,d\in X$, $c,e\in Y$, and $w(c)\neq w(e)$.
    \end{itemize}
\end{definition}

\begin{theorem}\label{alg-Dom3}
Let $G = (C \cup I, E)$ be a $K_{1,4}$-free $3$-split graph. Let $x \in C$ such that $d_G^I(x) = 3$. Let $H = (C \cup \hat{I}, \hat{E})$ be the graph obtained from $G$ by removing the set $N_G^I(x)$, where
 $\hat{I} = I \setminus N_G^I(x)$ and $\hat{E} = E \setminus \{xy \mid y \in N_G^I(x)\}$, and let $L$ be the labeled clique vertex set of the graph $H$. Construct the reduced graph $H'=(C\cup I',E')$, where $I'=\hat{I}\setminus N_{H}^{\hat{I}}(L), 
E'=\hat{E}\setminus \{l_v i_v \mid l_v\in L,\, i_v\in N_H^{\hat{I}}(L)\}$. Form a bipartite graph $Q=(A\cup B, E_Q)$, where $A=N_G^{I}(x)=\{x_1,x_2,x_3\}$, $B=L\cup N_{H'}(I')$, and $E_Q=\{c x_j, l_mx_j\mid l_m\in (N_G(x_j)\cap L), c\in (N_G(x_j)\cap  N_{H'}(I')), j\in [3]\}$ and $w
    (c_i)=w(c_j)$ if $N_{H'}^{I'}(c_i)\cap N_{H'}^{I'}(c_j)\neq \emptyset$, for any $c_i,c_j\in B$.\\
Then, the $m$-eternal domination number, \\ $\gamma_m^\infty(G) = 
\begin{cases}
|L| + |I'| + 1, & \text{if the graph $Q$ is either of Type I or Type II}, \\
|L| + |I'| + 2, & \text{otherwise.}
\end{cases}
$
\end{theorem}
\begin{proof}

Given a $K_{1,4}$-free $3$-split graph $G$, we have $\Delta_G^I=3$. Let $x\in C$ such that $d_G^I(x)=3.$ 
Let $H$ be the resulting graph, after removing $N_G^I(x)$ from $G$. From Lemma \ref{lemlsplit}, it follows that $H$ is a $K_{1,4}$-free $l$-split graph where $0\le l\le 2$. Observe that if $l\le 1$, then $L=\emptyset$. In particular, if $l=0$, then $x$ is a universal vertex of $G$. From Observation \ref{universal}, $\gamma_m^\infty(G)=2=|L|+|I'|+2$, where $L$ and $I'$ are empty set, and $Q$ is neither of \emph{Type I} nor of \emph{Type II}. Because $\mathcal{M}$ is the maximum matching of the labeled graph $M$, so, for each vertex $l_v\in L$, there exist at least a vertex $v_k\in N^{\hat{I}}(l_v)$ such that $N(v_k)\cap N(I')=\emptyset$. Now, let $U$ be a set that picks one such vertex from each $N(l_v)$. Observe that $N_G(v_i)\cap N_G(v_j)=\emptyset$ for any $v_i,v_i\in U\cup I'$. Thus, by Observation \ref{maximal_independnet}, we have $\gamma(G)\ge |U|+|I'|=|L|+|I'|$. Further, $\gamma_m^\infty(G)\ge |L|+|I'|$ follows from Observation \ref{o1}. Now, we consider two cases based on the graph $Q$ (note that $l=\Delta_H^{\hat{I}}\le 2$).
    \begin{itemize}
        \item \emph{Case 1}: (If the graph $Q$ is either of \emph{Type I} or \emph{Type II}) Now, we show $\gamma_m^\infty(G)>|L|+|I'|$. Let $D'$ be a minimal $m$-eternal dominating set of $G$ such that $|D'|=|L|+|I'|$. Because there exists a vertex $x\in C\setminus (N(I')\cup L)$ that has to be dominated by $D'$, so $D'\cap (N(I')\cup L) \neq \emptyset$. An attack on $x$ can not be defended as moving any vertex from $D'\cap (N(I')\cup L)$ (say $c_j$) leaves $v_j\in ((I'\cup U)\cap N(c_j))$ undominated. Hence, $\gamma_m^\infty(G)\ge |I'|+|L|+1$. Now, let $C_i=N(v_i), v_i\in I'$ and $D=\{c_i\mid c_i$\ is the representative of $C_i\}\cup L \cup \{i^*\}$ where $i^*\in I$. Clearly, $|D|=|I'|+|L|+1$. Next, we show $D$ is an $m$-eternal dominating set of $G$. $D$ is a dominating set of $G$ because $N[\{c_i\mid c_i$\ is the representative of $C_i\}\cup L]=V$. Define $D_1=D$ and $D_2=D\setminus \{i^*\}\cup \{c^*\}$, where $c^*\in C\setminus [N_G(I')\cup L]$. By maintaining one guard in each representative vertex of $C_i$ and the set $L$ as invariant, we defend the attack on any vertex of $G$ using a strategy shown in Table \ref{tab:K_1,4_3-SG}. Therefore, $\gamma_m^\infty(G)=|I'|+|L|+1$. 
    
        \item \emph{Case 2}: (If the graph $Q$ is neither of \emph{Type I} nor of \emph{Type II}) Then there exist a vertex $u\in N^I(x)$ such that $N(u)\cap (\{c_i\mid c_i$\ is the representative of $C_i\}\cup L)=\emptyset$. So, to dominate $u$, $\gamma(G)\ge |L|+|I'|+1$. This implies, $\gamma_m^\infty(G)\ge |L|+|I'|+1$. Now, we claim that $\gamma_m^\infty(G) > |L|+|I'|+1$. Let $D'$ be a minimal $m$-eternal dominating set of $G$ such that $|D'|=|I'|+|L|+1$. In fact, $D'\cap \{x,u\} \neq \emptyset$, as $D'$ is also a dominating set of $G$. Then an attack on vertex $v\in (N_G^I(x)\setminus \{u\})$ can not be defended as moving any of its adjacent guards (say $c_j$ (resp. $x$)) to $v$ leaves $v_j\in (N(c_j)\cap (I'\cup U)) (resp. $u$)$ undominated. Hence, $\gamma_m^\infty(G)\ge |I'|+|L|+2$. 
    Now, let $D=\{c_i\mid c_i$\ is the representative of $C_i\}\cup L \cup \{i^*,x\}$ where $i^*\in I$. Next, we show $D$ is an $m$-eternal dominating set of $G$. $D$ is a dominating set of $G$ because $N[\{c_i\mid c_i$\ is the representative of $C_i\}\cup L]=V$. Define $D_1=D$ and $D_2=(D\setminus \{i^*\})\cup \{c^*\}$, where $c^*\in C\setminus (N_G(I')\cup L)$. By maintaining one guard at each representative vertex of $C_i$ and the set $L\cup \{x\}$ as an invariant, the attack on any vertex of $G$ can be defended by the strategy shown in Table \ref{tab:K_1,4_3-SG}. Therefore, $\gamma_m^\infty(G)= |I'|+|L|+2$. 
    \end{itemize}

    \begin{table}[h!]
     \centering
    \renewcommand{\arraystretch}{1.3}
    \begin{tabular}{|c|>{\centering\arraybackslash}p{6cm}|>{\centering\arraybackslash}p{3cm}|}
    \hline
    \multirow{2}{3 cm}{Vertex on which attack happened} & 
    \multicolumn{2}{c|}{Type of guard configuration at the time of attack} \\
    \cline{2-3}
    & $D_1$ & $D_2$ \\
         \hline
         $\{c_v, c_v\in C\setminus \mathcal{B}\}$& $D_2 \{v^*\rightarrow c_v, i^*\rightarrow v^*\}, ~\text{where}~ v^*\in N(i^*)\cap D_1$  & $D_2 ~\{c^*\rightarrow c_v\}$\\
         \hline
        $\{i_v, i_v\in I\}$ & $D_1 ~\{v_i\rightarrow i_v, v^*\rightarrow v_i, i^*\rightarrow v^*\}  ~\text{where}~ v_i\in N(i_v)\cap D_1$ & $D_1 \{v_i\rightarrow i_v, c^*\rightarrow v_i\}$\\
         \hline
    \end{tabular}
    \caption{\centering Defending strategies against vertex attacks, illustrated for Cases 2 and 3. Here $\mathcal{B}$ represents the invariant taken in each corresponding subcase.}
    \label{tab:K_1,4_3-SG}
\end{table}
\end{proof}

\emph{Time Complexity:} Assume that the graph $G=(C\cup I, E)$ is represented using an adjacency lists. Identifying the vertex $x$ such that $d_G(x)=3$ requires scanning all vertices of $C$, which takes $O(|V|)$ time. Constructing the graph $H=(\hat{C}\cup \hat{I}, \hat{E})$ from $G$ by removing the neighborhood $N_G^I(x)$ and all its incident edges can also be done in $O(|V|)$ time, since only the adjacency lists of $x$ and its three neighbors are updated. Next, computing $\Delta_{H}^{\hat{I}}$ requires visiting each vertex in $\hat{C}$ and counting its neighbors, in $\hat{I}$. Because the graph is a $l$-split graph ($l\le 2$), this step takes $O(|V|)$ time. Forming the graph $H'$ can be done in $O(|V|)$ time because by Lemma \ref{lemlsplit}, the maximum matching $\mathcal{M}\le 2$ of the labeled graph $M$ of $H$, so by iterating over degree 2 vertices of $C$, $L$ can be formed. Constructing the bipartite graph $Q=(A\cup B, E_Q)$ requires only examining the adjacency relation between the constant size set $A$ and the vertices in $B$; hence, these steps also take $O(|V|)$ time. Checking whether the auxiliary graph $Q$ is of \emph{Type I} or \emph{Type II} can also be performed in constant time by checking if the vertices of $A$ have degree at least 1 in $Q$. Hence, the overall time taken by the algorithm derived from Theorem \ref{alg-Dom3} is $O(|V|)$ time.

\subsection{$K_{1,5}$-free Split Graphs}
\label{sec:$K_{1,5}$-free-SG}
\vspace{-5pt}
The classical problem \textsc{Exact 3-Cover} is defined as: given a collection $\mathcal{C}=\{C_1,C_2,\ldots,C_m\}$ of 3-element subsets of a set $\mathcal{Z}=\{u_1,u_2,\ldots,u_{3q}\}$, the aim is to find if there exists a sub-collection $\mathcal{C}'\subseteq \mathcal{C}$ such that for every $u_i\in \mathcal{Z}, i\in [3q]$, is in exactly one member of $\mathcal{C}'.$ 
\vspace{-5pt}
\begin{theorem} \cite{karp2009reducibility}
    The \textsc{Exact 3-Cover} is \textsf{NP}-complete.
\end{theorem}
\vspace{-5pt}
\begin{theorem}\label{K_1,5_NP_hard}
$m$-\textsc{Eternal Dominating Set} is \textsf{NP}-complete for $K_{1,5}$-free splits graphs. 
\end{theorem}

\begin{proof}
It has been shown that there exists a succinct certificate for the existence of an $m$-eternal dominating set of a given size \cite{macgillivrayeternal}. Therefore, for split graphs, the 
$m$-eternal domination problem is in \textsf{NP}.
Given an instance of exact 3-cover $(X,\mathcal{C})$, where $X=\{x_1,x_2\ldots,x_{3q}\}$ and $\mathcal{C}=\{C_1,C_2,\ldots,C_m\}$, $C_i\subseteq X,$ for all $i\in[m]$, we construct the graph $G=(V,E)$ as an instance of $m$-\textsc{Eternal Dominating Set}, in the following manner. Corresponding to each $C_i\in \mathcal{C}$, create a vertex $c_i\in V$ and corresponding to each element $x_i\in X,$ create a vertex $x_i\in V.$ Make $c_i$ adjacent to $x_j$ if and only if $x_j\in C_i.$ Also, make the vertex set $\{c_i\mid i\in [m]\}$ such that $G[\{c_i\mid i\in [m]\}]$ is a clique. Finally, add three new vertices namely, $u$, $v$, and $w$ such that $u$ is adjacent to each $c_i$, $i\in [m]$ and $v$ and $w$ are adjacent to their unique neighbor $u$.
    This completes the construction of the graph $G$. The corresponding construction is illustrated in a figure \ref{fig:split}. 

    \begin{figure}[htbp]
    \centering
    \includegraphics[width=3cm, height=4 cm]{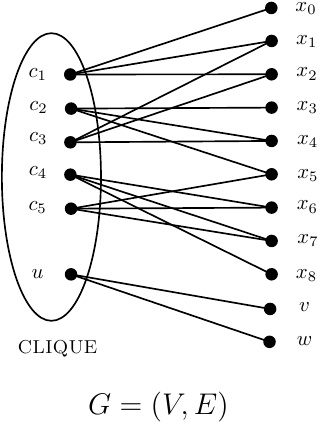}
    \caption{Illustration of construction of $G$ from an instance of exact 3-cover $(X,\mathcal{C})$, where $C_1=\{x_0,x_1,x_2\}$, $C_2=\{x_3,x_4,x_5\}$, $C_3=\{x_1,x_2,x_4\}$, $C_4=\{x_6,x_7,x_8\}$, and $C_5=\{x_5,x_6,x_7\}$ in Theorem \ref{K_1,5_NP_hard}.}
    \label{fig:split}
\end{figure}

Observe that $G$ is a split graph with split partition of vertex set $V=C\cup I$, where the clique $C=\{c_i\mid i\in [m]\}\cup \{u\}$, the independent set $I=\{x_i\mid i \in [3q]\}\cup \{v,w\}$, and the edge set $E=\{c_ix_j\mid x_j\in C_i, j\in[3q], i\in [m]\}\cup \{uc_i\mid i\in [m]\}\cup \{uv,uw\}$. In fact, it is a $K_{1,5}$-free split graph, since each vertex in $C$ has at most three neighbors in $I$. Also, it is easy to see that the construction can be done in polynomial time. Now we claim the following. 
\vspace{-5pt}
\begin{claim}\label{claim-K_1,5-NPhard}
$(X,\mathcal{C})$ has an exact cover $\mathcal{C'}$ of size $q$ if and only if $G$ has an $m$-eternal dominating set $S$ of cardinality at most $q+2.$
\end{claim}
\begin{proof}
Let $\mathcal{C'}$ be an exact cover of $(X,\mathcal{C})$ of cardinality $q.$ Define $S=D\cup \{u,w\},$ where $D=\{c_i\mid C_i\in \mathcal{C'}\}$. Clearly, $|S|\le q+2$. Now, we show that $S$ is an $m$-eternal dominating set of $G$. For that, we must be able to defend against all attacks. Let $S^*=D\cup \{u\}$. Observe that $S^*$ is a dominating set of $G$. Further, let the two types of family of sets as $\mathcal{S}_1=\{S^*\cup \{\beta\}\mid  \beta\in \{v,w\}\} ~\text{and}~ \mathcal{S}_2=\{S^*\cup \{\delta\} \mid \delta\in\{c^*,x^*\}\}, ~\text{where}~ c^*\notin S^*, x^*\in \{x_1,x_2,\ldots,x_{3q}\}$. Note that $S\in \mathcal{S}_1$. Since $S^*\subseteq \hat{S}, ~\text{for any}~ \hat{S}\in \mathcal{S}_i, i\in [2]$, this follows that $\mathcal{S}_1$ and $\mathcal{S}_2$ consists of dominating sets of $G$. 
The attack on a vertex can be defended by following a strategy such that the new configuration of vertices also forms a dominating set of $G$, that is, either of type $\mathcal{S}_1$ or type $\mathcal{S}_2$. We maintain $S^*$ as the invariant while defending against any attack.

\begin{table}[h!]
     \centering
    \renewcommand{\arraystretch}{1.3}
    \begin{tabular}{|c|>{\centering\arraybackslash}p{2.5cm}|>{\centering\arraybackslash}p{6.5cm}|}
    \hline
    \multirow{2}{3cm}{Vertex on which attack happened} & 
    \multicolumn{2}{c|}{Type of guard configuration at the time of attack} \\
    \cline{2-3}
    & $\mathcal{S}_1$ & $\mathcal{S}_2$ \\
         \hline
         $\{x_i, i\in [3q]\}$& $\mathcal{S}_2$ $\{c_{x_i}\rightarrow x_i, u\rightarrow c_{x_i}, \beta\rightarrow u\}$ & $\mathcal{S}_2$ $\begin{cases}
           c_{x_i}\rightarrow x_i, c^*\rightarrow c_{x_i}, &\text{if}~ \delta=c^* \\
           c_{x_i}\rightarrow x_i, c_{x^*}\rightarrow c_{x_i}, x^*\rightarrow  c_{x^*}. &\text{if}~ \delta=x^*
         \end{cases}$\\
         \hline
        $\{c_k, c_k\notin S^*\}$ & $\mathcal{S}_2$ $\{u\rightarrow c_k, v(w)\rightarrow u\}$ & $\mathcal{S}_2$ $\begin{cases}
           c^*\rightarrow c_k, &\text{if}~ \delta=c^* \\
           c_{x^*}\rightarrow c_k, x^*\rightarrow  c_{x^*}, &\text{if}~ \delta=x^*
         \end{cases}$\\
         \hline
        $v(w)$ & $\mathcal{S}_1$ $\{u\rightarrow v(w), w(v)\rightarrow u\}$ & $\mathcal{S}_1$ $\begin{cases}
           u\rightarrow v(w), c^*\rightarrow u, &\text{if}~ \delta=c^* \\
           u\rightarrow v(w), c_{x^*}\rightarrow u, x^*\rightarrow  c_{x^*}, &\text{if}~ \delta=x^*
         \end{cases}$ \\
         \hline
    \end{tabular}
    \caption{The cell in row 2 under column 2 indicates the type of configuration the guards are initially in. Each cell in rows 3, 4, and 5 under the same column represents the configuration to which the guards move, and the text in brackets indicates the corresponding strategy. Here $c_{x_i}\in S^*\cap N(x_i)$, $c_{x^*}\in N(x^*)\cap S^*$.}
    \label{tab:my_label}
\end{table}

Conversely, let $S$ be an $m$-eternal dominating set of $G$ of cardinality at most $q+2$. We first show that any $m$-eternal dominating set of $G$ of cardinality at most $q+2$ intersects the set $X=\{x_1, x_2,\ldots, x_{3q}\}$ in at most one vertex.

\begin{claim}
    For any $m$-eternal dominating set $S$ of $G$ of size at most $q+2$, $|X\cap S|\le 1$.
\end{claim}
\begin{proof}
    On the contrary, let $|X\cap S|=t$ where $t\ge 2$. Since, $S$ is a dominating set of $G$, $|S\cap \{u,v,w\}|\ge 1$. Observe that each vertex of $C\setminus \{u\}$ can defend at most three vertices of $X$. Thus, to defend the remaining $3q-t$ vertices of $X$ using at most $q-t+1$ vertices of $S$, it hold that $3q-t\le 3(q-t+1)$ that implies $t\le \frac{3}{2}$, a contradiction as $t\ge 2$.    
\end{proof}

Further, in the following claim, we show that any $m$-eternal dominating set of $G$ of cardinality at most $q+2$, intersects the set $\{u,v,w\}$ in at most two vertices.

\begin{claim}\label{claim2}
    For any $m$-eternal dominating set $S$ of $G$ of size at most $q+2$, $|S\cap \{u,v,w\}|\le 2$. Moreover, if $|S\cap \{u,v,w\}|=2$ then $|X\cap S|=0$.
\end{claim}
\begin{proof}
    Assume for contradiction that $|S\cap \{u,v,w\}|\neq 2$. That is $|S\cap \{u,v,w\}|= 3$ implies $\{u,v,w\}\subseteq S$. First, consider the case $|X\cap S|=0$. To defend $3q$ vertices of $X$ using at most $q-3$ vertices of $S$, it hold that $3q\le 3(q-3)$ that implies $0 \le -9 $, which is absurd. Now, consider the case $|X\cap S|=1$. To defend the remaining $3q-1$ vertices of $X$ using at most $q-4$ vertices of $S$, it hold that $3q-1\le 3(q-4)$ that implies $1\ge 12 $, which is absurd. 

    Next, let $|S\cap \{u,v,w\}|=2$ but $|X\cap S|=p$ where $p\ge 1$. To defend the remaining $3q-p$ vertices of $X$ using at most $q-p$ vertices of $S$, it holds that $3q-p\le 3(q-p)$, which implies $1\ge 3$, which is absurd. 
\end{proof}

Now, define $\hat{C}=\{c_i\in (S\cap C) \mid |S\cap \{u,v,w\}|=2\}$ and $\mathcal{C}'=\{C_i\mid c_i\in \hat{C}\}$. As $N(\{u,v,w\})\cap X=\emptyset$ and each vertex $c_i\in C$, $i\in [m]$ has at most 3 neighbors in $X$, to dominate all $3q$ elements of $X$, $|S\cap \{c_i\mid i\in [m]\}|\ge q$. Now, we claim that $\mathcal{C'}$ is an exact cover of $(X,\mathcal{C}).$ From Claim \ref{claim2}, $|X\cap S|=0$ and as $S$ is a dominating set of $G$, it follows that each vertex of $X$ must be dominated by at least one vertex of $\hat{C}$. Hence, $|\mathcal{C'}|=|S|-2\le q$ that implies $|\mathcal{C}'|=q$.
\end{proof}

Therefore, $m$-\textsc{Eternal Dominating Set} is \textsf{NP}-hard for $K_{1,5}$-free splits graphs. 
\end{proof}
\vspace{-5pt}
Since every $K_{1,t}$-free split graphs is also $K_{1,5}$-free splits graphs for $t\ge 5$, we have the following corollary from Theorem \ref{K_1,5_NP_hard}.
\begin{corollary}
   $m$-\textsc{Eternal Dominating Set} is \textsf{NP}-complete for $K_{1,t}$-free splits graphs for $t\ge 5$.  
\end{corollary}
\vspace{-15pt}
\section{Undirected Path Graph}
\label{sec:UPG}
\vspace{-10pt}

This section shows that the \textsc{$m$-Eternal dominating Set} is \textsf{NP}-hard for undirected path graphs, a subclass of chordal graphs. For this, we provide a polynomial-time reduction from the well-known $3D$-matching problem.

The $3D$-matching problem is defined as follows: let $W, X$, and $Y$ be three disjoint sets, each of cardinality $q$, and let $M$ be a subset of $W\times X\times Y$ of cardinality $p$. The \emph{$3D$-matching problem} asks if there exists a subset $M'$ of $M$ of cardinality exactly $q$ such that each element of $W, X$ and $Y$ occurs exactly once in a triple of $M'$.
\vspace{-5pt}
\begin{theorem}\cite{garey2002computers}
The $3D$-matching problem is \textsf{NP}-complete. 
\end{theorem}
\vspace{-10pt}
\begin{theorem}\label{NP_UPG}
$m$-\textsc{Eternal Dominating Set} is \textsf{NP}-hard for undirected path graphs. 
\end{theorem}
\vspace{-10pt}
\begin{proof}
Let $W, X$, and $Y$ be three disjoint sets, each of cardinality $q$, and let $M$ be a subset of $W\times X\times Y$ having cardinality $p$, as an instance of the $3D$-matching problem. For notation, we use $W=\{w_j\mid j\in [q]\}$, $X=\{x_k\mid k\in [q]\}$, $Y=\{y_l\mid l\in [q]\}$ and $M=\{m_i=(w_j,x_k,y_l)\mid w_j\in W, x_k\in X, y_l\in Y,i\in [p]\}$. Without loss of generality, we can assume each element of $W$, $X$, and $Y$ occurs in at least two triples of $M$. Using ideas similar to the construction in \cite{booth1982dominating}, from the $3D$-matching problem, we construct a graph $G$ as an instance of $m$-\textsc{Eternal Dominating Set} in the following way. 
 For each $m_i\in M$, create a gadget $M_i$, $i\in[p]$ (graph induced by the vertex set $\{a_i,b_i,c_i,d_i,e_i,f_i,g_i,h_i,i_i\}$ as shown in Figure \ref{fig:UPGgraph}). For each $w_j\in W$, $x_k\in X$ and $y_l\in Y$, create a vertex $w_j$, $x_k$, and $y_l$, respectively. Add edges $w_ja_i$, $x_kb_i$, and $y_lc_i$ if and only if $w_j\in m_i$, $x_k\in m_i$, and $y_l\in m_i$ respectively, for all $j,k,l\in [q], i\in[p]$. Further, make the graph induced by $\{a_i,b_i,c_i\mid i\in [p]\}$ a clique. Finally, introduce three new vertices $u$, $v$, and $w$ such that $u$ is adjacent to all $a_i,b_i,c_i, i\in [p]$ and $v,w$ are adjacent to their unique neighbor $u$.

To clarify the above construction, an illustration of the graph $G$ construction from a $3D$-matching problem instance is provided in Figure \ref{fig:UPGgraph}.

     \begin{figure}[ht]
        \centering
        \includegraphics[width=10 cm, height=5 cm]{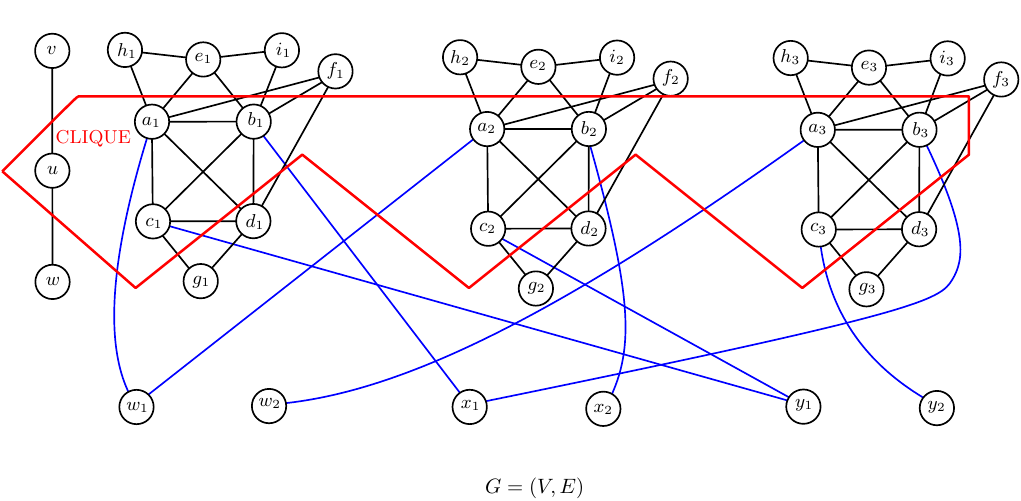}
        \caption{The graph above corresponds to the instance $W=\{w_1,w_2\}, X=\{x_1,x_2\},Y=\{y_1,y_2\}$ and $M=\{m_1=(w_1,x_1,y_1),m_2=(w_1,x_2,y_1),m_3=(w_2,x_1,y_2)\}$ of the $3D$-matching problem in accordance with the construction of the Theorem \ref{NP_UPG}.}
        \label{fig:UPGgraph}
    \end{figure}

    First, \emph{we show that the constructed graph $G$ is an undirected path graph.} For this, we need to show that there exists a clique tree $T$ (recall that vertices of the clique tree correspond to maximal cliques in $G$), such that for each vertex $v_i\in V(G)$, the set of all maximal cliques of $G$ containing $v_i$, i.e., $P_{v_i}=\{C_j\mid v_i\in V(C_j)\}$ induces a path in $T$. The clique tree $T$ associated with the graph $G$ is illustrated in Figure \ref{clique tree}. It is clear from Figure \ref{clique tree} that $P_{v_i}$, $P_{s_t},$ where $v_i\in\{a_i,b_i,c_i,d_i,e_i,f_i,g_i,h_i,i_i\mid i\in [p]\}$ and $s_t\in \{w_j,x_k,y_l\mid j,k,l\in [q]\}$, and $P_u, P_v, P_w$ induces a path in the clique tree $T$. 

     \begin{figure}[h!]
        \centering
        \includegraphics[width=12 cm, height=5.5cm]{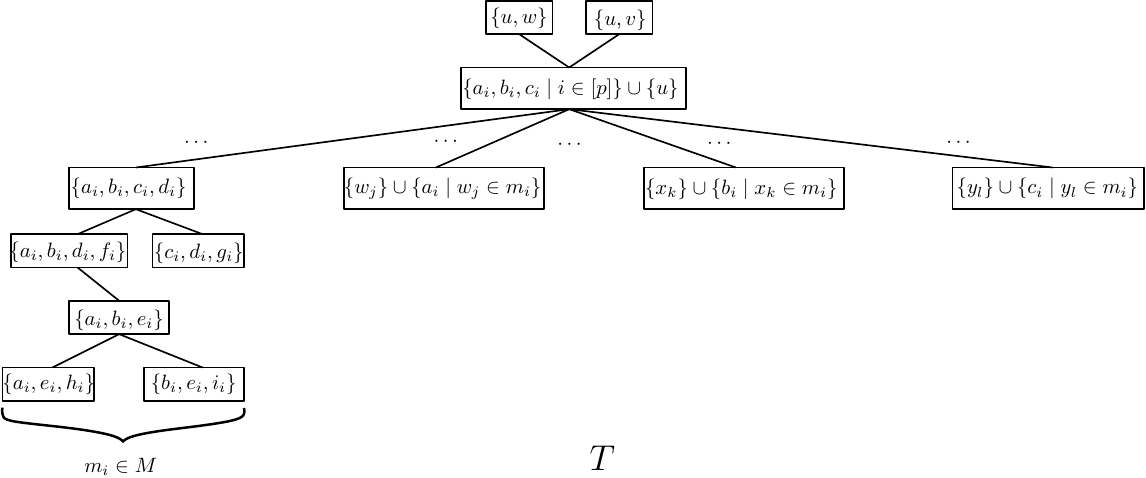}
        \caption{The clique tree $T$ corresponding to the graph $G$ constructed from $3D$-matching problem in Theorem \ref{NP_UPG}.}
        \label{clique tree}
    \end{figure}
    
Next, \emph{we claim that the $3D$-matching problem $M$ has a solution $M'$ of size $q$ if and only if the graph $G$ has an $m$-eternal dominating set $D$ of cardinality at most $2p+q+2$.}
    \begin{claim}
      If $M$ has a solution $M'$ of size $q$, then the graph $G$ has an $m$-eternal dominating set $D$ of cardinality at most $2p+q+2$.    
    \end{claim}
    \begin{proof}
    Let $M'$ be the solution of $M$. Define $D=\{a_i,b_i,c_i\mid m_i\in M'\} \cup \{d_i,e_i\mid m_i\notin M'\}\cup \{u,v\}$. Clearly, $|D|=3q+2(p-q)+2=2p+q+2$ and $D$ is a dominating set of $G$. Now, we show that $D$ is an $m$-eternal dominating set of $G$. 
    
    We claim that \emph{$D\subseteq V$ is an $m$-eternal dominating set of $G$}.

 For this, we need to show that starting with the set $D$, we can defend against any attack eternally. Now, let $D'=\{a_i,b_i,c_i\mid m_i\in M'\} \cup \{u\}, A=\{w_j,x_k,y_l \mid j,k,l\in [q]\}$ and further, let us define following family $\mathcal{F_V}=\{\mathcal{D}_1,\mathcal{D}_2,\mathcal{D}_3\}$ consisting of subsets of $V$. Define $\beta\in \{v,w\}$, $\theta \in \{d_i,e_i,f_i,g_i,h_i,i_i\mid m_i\in M'\}$, $\epsilon \in \{a_i,b_i,c_i\mid m_i\notin M'\}$, and $s_t\in A$.\\ Define $\mathcal{D}_1 = \left\{D_1^\beta= \left\{D' \cup \{d_i,e_i\mid m_i\notin M'\}\cup \{\beta\}\right\}\right\}$, \\$\mathcal{D}_2 = \left\{
       D_1^\beta \setminus \{\beta\} \cup \{s_t\},
        D_1^\beta \setminus \{\beta\} \cup \{\theta\},
         D_1^\beta \setminus \{\beta\} \cup \{\epsilon\}\right\},$\\
         and $\mathcal{D}_3 = \left\{\begin{array}{l}
       D_1^\beta \setminus \{\beta,d_i\} \cup \{f_i,c_i\},
        D_1^\beta \setminus \{\beta,d_i\} \cup \{g_i,b_i\},\\
         D_1^\beta \setminus \{\beta,e_i\} \cup \{h_i,b_i\},  D_1^\beta \setminus \{\beta,e_i\} \cup \{i_i,a_i\}
        \end{array}\right\},\\ \text{where}~ \{d_i,e_i,f_i,c_i,a_i,b_i\mid m_i\notin M'\}.$

Observe that for any $D^*\in \mathcal{D}_i , D^{**}\in \mathcal{D}_j, i,j\in [3]$, $|D^*|=|D^{**}|$. Now, we show that $\mathcal{D}_i, i\in [3]$ consists of subsets of $V$ that are dominating sets of $G$.

Firstly, as $D'\subseteq D^*$, for all $D^*\in \mathcal{D}_i$, $i\in [3]$, and $D'$ dominates $\hat{V}=\{V(M_i)\mid m_i\in M'\}\cup A \cup \{u,v,w\} \cup \{a_i,b_i,c_i\mid m_i\notin M'\}$. In other words, $\hat{V}=V\setminus \{d_i,e_i,f_i,g_i,h_i,i_i\mid m_i\notin M'\}$. Now, to prove each $D^*\in \mathcal{D}_i, i\in [3]$ is a dominating set of $G$, we only need to show the remaining vertex set $V\setminus \hat{V}=\{d_i,e_i,f_i,g_i,h_i,i_i\mid m_i\notin M'\}$ is dominated by $D^*$. 

For any set $D^*\in \mathcal{D}_1$, as $\{d_i,e_i\mid m_i\notin M'\}\subseteq D^*$, it dominates $\{V(M_i)\mid m_i\notin M'\}$. Hence, $D^*$ is a dominating set of $G$. 
  
For any set $D^*\in \mathcal{D}_2$, as $\beta\in N(u), u\in \mathcal{D}_1^\beta$, and $\mathcal{D}_1^\beta$ is a dominating set of $G$, implies $\mathcal{D}_1^\beta \setminus \beta$ is also a dominating set of $G$. As $\mathcal{D}_1^\beta \setminus \beta \subseteq D^*$, for all $D^*\in \mathcal{D}_2$, hence, $D^*$ is a dominating set of $G$. 

For any set $D^*\in \mathcal{D}_3$, as one of $\{f_i,c_i,e_i\},\{g_i,b_i,e_i\},\{b_i,h_i,d_i\},\{i_i,a_i,d_i\}\subseteq D^*$ in each set $D^*\in \mathcal{D}_3$, and it dominates $\{V(M_i)\mid m_i\notin M_i\}$. Hence, $D^*$ is a dominating set of $G$.

We show that $D = D_1^{\beta}$ is an $m$-eternal dominating set of $G$. We have already shown that each set family in $\mathcal{F_V}$ consists of dominating sets of $G$. To show that $D$ is an $m$-eternal dominating set, for each set in $D_i \in \mathcal{D}_i$, $i\in [3]$, we show that any attack can be defended, and that the resultant dominating set also belongs to one of the set families in $\mathcal{F_V}$. In Table \ref{tab:my_label-o}, we illustrate the types of guard configurations before and after the attacks.    

\begin{table}[h!]
    \centering
    \renewcommand{\arraystretch}{1}
    \begin{tabular}{|c|>{\centering\arraybackslash}p{2.5cm}|>{\centering\arraybackslash}p{2.5cm}|>{\centering\arraybackslash}p{2.5cm}|}
    \hline
    \multirow{2}{4.5 cm}{Vertex on which the attack happened} & 
    \multicolumn{3}{c|}{Type of guard configuration at the time of attack} \\
    \cline{2-4}
    & $\mathcal{D}_1$ & $\mathcal{D}_2$ & $\mathcal{D}_3$ \\
    \hline
    $\{\beta\mid \beta\in \{v,w\}\}$ & $\mathcal{D}_1$ & $\mathcal{D}_1$ & $\mathcal{D}_1$\\
    \hline
    $\{a_i,b_i,c_i\mid m_i\notin M_i\}$ & $\mathcal{D}_2$  & $\mathcal{D}_2$ & $\mathcal{D}_2$\\
    \hline
    $\{d_i,e_i\mid m_i\notin M_i\}$ & $-$ & $-$ & $\mathcal{D}_1$\\
    \hline
    $\{f_i,g_i,h_i,i_i\mid m_i\notin M_i\}$ & $\mathcal{D}_3$ & $\mathcal{D}_3$ & $\mathcal{D}_3$\\
    \hline
    $\{d_i,e_i,f_i,g_i,h_i,i_i\mid m_i\in M_i\}$ & $\mathcal{D}_2$& $\mathcal{D}_2$ & $\mathcal{D}_2$ \\
    \hline
    \end{tabular}
    \caption{The cell in row 1 under column 2 indicates the type of configuration the guards are initially in. Each cell in row $i, 2\le i\le 5$ under column 2 represents the configuration the guards move to when the attack happened at the vertex in cell row $i$ under column 1. ``–'' indicates no change in the configuration, as the attacked vertex is already occupied in the current configuration.}
    \label{tab:my_label-o}
\end{table}

The detailed defending strategy for an attack can be found in Table \ref{tab:my_label2}.

\begin{table}[h!]
    \centering
    \renewcommand{\arraystretch}{1}
    \begin{tabular}{|c|>{\centering\arraybackslash}p{2cm}|>{\centering\arraybackslash}p{3cm}|>{\centering\arraybackslash}p{4cm}|}
    \hline
    \multirow{2}{3 cm}{Vertex on which the attack happened} & 
    \multicolumn{3}{c|}{Type of guard configuration at the time of attack} \\
    \cline{2-4}
    & $\mathcal{D}_1$ & $\mathcal{D}_2$ & $\mathcal{D}_3$ \\
    \hline
    $v(w)$ & $\mathcal{D}_1$ $\{u\rightarrow v(w), w(v)\rightarrow u\}$ & $\mathcal{D}_1$ $\{u\rightarrow v(w), x_{s_t}\rightarrow u, s_t\rightarrow x_{s_t}: u\rightarrow v(w), z_\theta\rightarrow u, \theta\rightarrow z_\theta: u\rightarrow v(w), \epsilon\rightarrow u\}$ & $\mathcal{D}_1$ $\{u\rightarrow v(w), c_i\rightarrow u, f_i\rightarrow d_i: u\rightarrow v(w), b_i\rightarrow u, g_i\rightarrow d_i: u\rightarrow v(w), b_i\rightarrow u, h_i\rightarrow e_i: u\rightarrow v(w), a_i\rightarrow u, i_i\rightarrow e_i\}$\\
    \hline
    $\epsilon^*$ & $\mathcal{D}_2$ $\{u\rightarrow \epsilon^*, \beta \rightarrow u\}$ & $\mathcal{D}_2$ $\{x_{s_t}\rightarrow \epsilon^*, s_t\rightarrow x_{s_t}: z_\theta\rightarrow \epsilon^*, \theta\rightarrow z_\theta: \epsilon\rightarrow \epsilon^* \}$ & $\mathcal{D}_2$ $\{c_i\rightarrow \epsilon^*, f_i\rightarrow d_i: b_i\rightarrow \epsilon^*, g_i\rightarrow d_i: b_i\rightarrow \epsilon^*, h_i\rightarrow e_i : a_i\rightarrow \epsilon^*, i_i\rightarrow e_i\}$\\
    \hline
    $\theta^*$ & $\mathcal{D}_2$ $\{z_{\theta^*}\rightarrow \theta^*, u\rightarrow z_{\theta^*}, \beta \rightarrow u\}$ & $\mathcal{D}_2$ $\{z_{\theta^*} \rightarrow \theta^*, x_{s_t}\rightarrow z_{\theta^*}, s_t\rightarrow x_{s_t}: z_{\theta^*} \rightarrow \theta^*, z_\theta\rightarrow z_{\theta^*},\theta\rightarrow z_\theta: z_{\theta^*}\rightarrow z_\theta, \epsilon\rightarrow z_{\theta^*} \}$ & $\mathcal{D}_2$ $\{z_{\theta^*}\rightarrow \theta^*, c_i\rightarrow z_{\theta^*}, f_i\rightarrow d_i: z_{\theta^*}\rightarrow \theta^*, b_i\rightarrow z_{\theta^*}, g_i\rightarrow d_i: z_{\theta^*}\rightarrow \theta^*, b_i\rightarrow z_{\theta^*}, h_i\rightarrow e_i: z_{\theta^*}\rightarrow \theta^*, a_i\rightarrow z_{\theta^*}, i_i\rightarrow e_i \}$\\
    \hline
    $d_i$ & $-$ & $-$ & $\mathcal{D}_1$ $\{f_i\rightarrow d_i, u\rightarrow \beta, c_i\rightarrow u: g_i\rightarrow d_i, u\rightarrow \beta, b_i\rightarrow u:-:- \}$\\
    \hline
    $e_i$ & $-$ & $-$ & $\mathcal{D}_1$ $\{-:-:h_i\rightarrow e_i, u\rightarrow \beta, b_i\rightarrow u: i_i\rightarrow e_i, u\rightarrow \beta, a_i\rightarrow u\}$\\
    \hline
    $f_i$ & $\mathcal{D}_3$ $\{d_i\rightarrow f_i, u\rightarrow c_i, \beta \rightarrow u\}$ & $\mathcal{D}_3$ $\{d_i\rightarrow f_i, x_{s_t}\rightarrow c_i, s_t\rightarrow x_{s_t}: d_i\rightarrow f_i, z_\theta\rightarrow c_i, \theta\rightarrow z_\theta: d_i\rightarrow f_i, \epsilon\rightarrow c_i\}$ & $\mathcal{D}_3$ $\{-:b_i\rightarrow f_i, g_i\rightarrow c_i:b_i\rightarrow f_i, d_i\rightarrow c_i, h_i\rightarrow e_i:d_i\rightarrow f_i, a_i\rightarrow c_i, i_i\rightarrow e_i\}$\\
    \hline
    $g_i$ & $\mathcal{D}_3$ $\{d_i\rightarrow g_i, u\rightarrow b_i, \beta \rightarrow u\}$ & $\mathcal{D}_3$ $\{d_i\rightarrow g_i, x_{s_t}\rightarrow b_i, s_t\rightarrow x_{s_t}: d_i\rightarrow g_i, z_\theta\rightarrow b_i, \theta\rightarrow z_\theta: d_i\rightarrow g_i, \epsilon\rightarrow b_i\}$ & $\mathcal{D}_3$ $\{c_i\rightarrow g_i, f_i\rightarrow b_i:-:d_i\rightarrow g_i, h_i\rightarrow e_i:d_i\rightarrow g_i, i_i\rightarrow b_i, a_i\rightarrow e_i\}$\\
    \hline
    $h_i$ & $\mathcal{D}_3$ $\{e_i\rightarrow h_i, u\rightarrow b_i, \beta \rightarrow u\}$ & $\mathcal{D}_3$ $\{e_i\rightarrow h_i, x_{s_t}\rightarrow b_i, s_t\rightarrow x_{s_t}: e_i\rightarrow h_i, z_\theta\rightarrow b_i, \theta\rightarrow z_\theta: e_i\rightarrow h_i, \epsilon\rightarrow b_i\}$ & $\mathcal{D}_3$ $\{e_i\rightarrow h_i, f_i\rightarrow b_i, c_i\rightarrow g_i:e_i\rightarrow h_i, g_i\rightarrow d_i:-:a_i\rightarrow h_i, i_i\rightarrow b_i\}$\\
    \hline
    $i_i$ & $\mathcal{D}_3$ $\{e_i\rightarrow i_i, u\rightarrow a_i, \beta \rightarrow u\}$ & $\mathcal{D}_3$ $\{e_i\rightarrow i_i, x_{s_t}\rightarrow a_i, s_t\rightarrow x_{s_t}: e_i\rightarrow i_i, z_\theta\rightarrow a_i, \theta\rightarrow z_\theta: e_i\rightarrow i_i, \epsilon\rightarrow a_i\}$ & $\mathcal{D}_3$ $\{e_i\rightarrow i_i, f_i\rightarrow a_i, c_i\rightarrow d_i:e_i\rightarrow i_i, b_i\rightarrow a_i, g_i\rightarrow d_i:b_i\rightarrow i_i, h_i\rightarrow a_i:-\}$\\
    \hline
    \end{tabular}
    \caption{The detailed defending strategy for each cell corresponding to Table \ref{tab:my_label-o}. Note that $x_{s_t}\in [N(s_t)\cap \{a_i,b_i,c_i\mid m_i\in M'\}]$, $z_\theta\in [N(\theta)\cap \{a_i,b_i,c_i\mid m_i\in M'\}]$, $\epsilon^*\in \{a_i,b_i,c_i\mid m_i\notin M'\}$, and $\theta^*\in \{d_i,e_i,f_i,g_i,h_i,i_i\mid m_i\in M'\}$.  Here ``–'' indicates no change in the configuration, as the attacked vertex is already occupied in the current configuration.}
    \label{tab:my_label2}
\end{table}

\end{proof}

    \begin{claim}\label{converse-UPG}
     If the graph $G$ has an $m$-eternal dominating set $D$ of cardinality at most $2p+q+2$, then the $3D$-matching instance $M$ has a solution $M'$ of size $q$.   
    \end{claim}

    \begin{proof}
    Let $D\subseteq V$ be an $m$-eternal dominating set of $G$ of cardinality at most $2p+q+2$. 
        
        First, we show that $D$ can intersect at most one vertex with the set of vertices $w_j$, $x_k$, $y_l$.
        \begin{claim}
            Let $A=\{w_j,x_k,y_l\mid j,k,l\in [q]\}$. Then $|D\cap A|\le 1$. Moreover, if $|D\cap A|=1$ then $|D\cap \{u,v,w\}|=1,$ otherwise $|D\cap \{u,v,w\}|=2$.
        \end{claim}
        \begin{proof}
             In contrast, let us assume $|D\cap A|=r\ge 2$. As $|D\cap M_i|\ge 2$ for each $i\in [p]$ and $|D\cap \{u,v,w\}|\ge 1$, to defend the remaining $3q-r$ vertices of $S$ we have at most $q-r+1$ gadgets $M_i$ such that $|D\cap V(M_i))|\ge 3$, which can defend at most $3(q-r+1)$ vertices of $A$. Hence, $3q-r\le 3(q-r+1)$, a contradiction to $r\ge 2$. Thus, $|D\cap A|\le 1$. 
             
             Next, we show if $|D\cap A|=1$ then $|D\cap \{u,v,w\}|=1.$ Suppose not, and $|D\cap \{u,v,w\}|=z\ge 2.$ To defend the remaining $3q-1$ vertices of $A$ we have at most $q-z+1$ gadgets $M_i$ such that $|D\cap V(M_i))|\ge 3$ that can defend at most $3(q-z+1)$ vertices of $A$. Hence, $3q-1\le 3(q-z+1)$, a contradiction to $z\ge 2$. 

           Now, it remains to show that if $|D\cap A|=0$ then $|D\cap \{u,v,w\}|=2.$ Suppose not, then $|D\cap \{u,v,w\}|\neq 2.$ Observe that $|D\cap \{u,v,w\}|\ge 1$ and if $|D\cap \{u,v,w\}|= 1$ then $u\in D$. Now, an attack at $v(w)$ can only be defended by moving the guard at $u$ to $v(w)$, and to dominate $w(v)$, some guard from the neighbors of $u$ must move to $u$. Without loss of generality, let the guard at $a_j\in M_j$ move to $u$. Note that, meanwhile, $|D\cap A|\neq 1$ as $|D\cap \{u,v,w\}|\neq 1$. Also $|D\cap \{u,v,w\}|\neq 2$ implies $|D\cap \{u,v,w\}|= 3.$ So, to defend the $3q$ vertices of $A$ we have at most $q-1$ gadgets $M_i$ such that $|D\cap V(M_i))|\ge 3$ that can defend at most $3(q-1)$ vertices of $A$. Hence, $3q\le 3(q-1)$, which is absurd.  
        \end{proof}
        Thus, we have $|D\cap \{\displaystyle\cup_{i\in p}^{} V(M_i)\}|\le 2p+q$. Moreover, $|D\cap \displaystyle\cup_{i\in p}^{} V(M_i)|\ge 2p$.
        \begin{claim}
            If $|D\cap A|\le 1$ then $|D\cap V(M_i)|\le 3$ for each $i\in [p]$.
        \end{claim}
        \begin{proof}
            Let's say there exists a gadget $M_i$, for some $i\in [p]$, such that $|D\cap M_i|=x>3$. Consider the case $|D\cap A|=1$. We know $|D\cap M_i|\ge 4$, and $|D\cap \cup_{j\in [p], j\neq i}^{} V(M_j)|\ge 2(p-1)$. As one gadget can defend at most 3 vertices of $A$ and hence, assuming $M_i$ can defend $y$ vertices of $A$ (note that $y\le 3$), $ 3q-y-1 \le 3(q-x+2) \implies 3x-6\le y+1,$ which is absurd. Now, consider the case $|D\cap A|=0$. Like the above case, $3q-y \le 3(q-x+2)\implies 3x-6\le y,$ which is again absurd. Hence, $|D\cap V(M_i)|\le 3$ for each $i\in [p]$, if $|D\cap A|\le 1$. 
        \end{proof}
    Now, we have $2\le |D\cap V(M_i)|\le 3$. Define $M'=\{m_i\in M\mid |D\cap M_i|=3\}$. Assume there are $t$ of $p$ gadgets $M_i$ such that $|D\cap M_i|=3$. Then $2p+q\ge 3t+2(p-t)\implies q\ge t$. Also, to defend $3q$ vertices of $A$ with the fact that each gadget can defend at most three unique vertices of $A$ implies there must exist at least $q$ gadgets of $M_i$ such that $|D\cap V(M_i)|=3$. Hence, $t=q$. 
    
    Finally, we show that $M'$ is a solution of $M$. For that, we show that each element of $W$, $X$, and $Y$ is in exactly one $m_i\in M'$. On contrary, assume there exists some element, say $w_i\in W$ (wlog) such that $w_i$ is in $\alpha$ number of $m_i$'s, such that $m_i\in M'$ and $\alpha\ge 2$. Since $D$ is an $m$-eternal dominating set of $G$, each $M_i$ can defend at most 3 vertices of $A$, and $|A|=3q$, implies that $3q \le 3(q-\alpha)+2\alpha+1 \implies \alpha\le 1$, a contradiction. 
    \end{proof}

    Therefore, the $3D$-matching problem $M$ has a solution $M'$ of size $q$ if and only if the graph $G$ has an $m$-eternal dominating set $D$ of cardinality at most $2p+q+2$.
\end{proof}

\section{Conclusion}
\vspace{-5pt}
In this paper, we established a dichotomy for the \textsc{$m$-Eternal Dominating Set} on $K_{1,t}$-free split graphs: it is solvable in polynomial time for 
$t\le 4$, and \textsf{NP}-complete for $t\ge 5$. We further proved that the problem remains \textsf{NP}-hard on undirected path graphs. Further, we established a complexity difference between \textsc{Dominating Set} and \textsc{Eternal Dominating Set} and also between \textsc{Eternal Dominating Set} and \textsc{$m$-Eternal Dominating Set}. Interestingly, we identified a graph class where \textsc{Dominating Set} is \textsf{NP}-hard, while \textsc{$m$-Eternal Dominating Set} is polynomial-time solvable. These results reveal contrasting complexity behaviors between domination and its eternal variant. Further, it would be interesting to design efficient graph algorithms for \textsc{$m$-Eternal Dominating Set} in other well-structured graph classes.

%
%
%

\bibliographystyle{splncs04}

\end{document}